\documentclass[journal, onecolumn,12pt, draftclsnofoot]{IEEEtran}
\usepackage{listings}
\usepackage{cite}
\usepackage{url}
\usepackage{xcolor}
\usepackage{graphicx}
\usepackage{epstopdf}
\usepackage[cmex10]{amsmath}
\usepackage{amsfonts}
\usepackage{amssymb}
\usepackage{amsthm}
\usepackage{mathrsfs}
\usepackage{verbatim}
\usepackage[noend]{algpseudocode}
\usepackage{algorithmicx,algorithm}
\usepackage{setspace}
\usepackage{bm}%
\usepackage{enumerate}
\theoremstyle{plain}
\newtheorem{theorem}{Theorem}
\newtheorem{lemma}[]{Lemma}
\newtheorem{corollary}{Corollary}

\usepackage[caption=false,font=normalsize,labelfont=sf,textfont=sf]{subfig}
\begin{document}
	\title{Optimizing Age of Information in Random-Access Poisson Networks}
	\author{Xinghua~Sun, \IEEEmembership{Member,~IEEE}, Fangming~Zhao, Howard H.~Yang, \IEEEmembership{Member,~IEEE}, Wen Zhan, \IEEEmembership{Member,~IEEE}, Xijun Wang, \IEEEmembership{Member,~IEEE}, \newline and Tony Q.~S.~Quek, \IEEEmembership{Fellow,~IEEE}
		\thanks{X. Sun, F. Zhao and W. Zhan are with the School of Electronics and Communication Engineering, Sun Yat-sen University, Guangdong, China (email: sunxinghua@mail.sysu.edu.cn, zhaofm@mail.sysu.edu.cn, zhanw6@mail.sysu.edu.cn).}
		\thanks{H. H. Yang is with the Zhejiang University/University of Illinois at Urbana-Champaign Institute, Zhejiang University, Haining 314400, China (e-mail:haoyang@intl.zju.edu.cn).}
		\thanks{X. Wang is with School of Electronics and Information Technology, Sun Yat-sen University, Guangdong, China (email: wangxijun@mail.sysu.edu.cn).}
		\thanks{T.\,Q.\,S.\,Quek is with the Information System Technology and Design Pillar, Singapore University of Technology and Design, Singapore 487372 (email: tonyquek@sutd.edu.sg).}}
	
	\maketitle
	\begin{abstract}
		Timeliness is an emerging requirement for many Internet of Things (IoT) applications. In IoT networks, where a large-number of nodes are distributed, severe interference may incur during the transmission phase which causes age of information (AoI) degradation. It is therefore important to study the performance limit of AoI as well as how to achieve such limit. In this paper, we aim to optimize the AoI in random access Poisson networks. By taking into account the spatio-temporal interactions amongst the transmitters, an expression of the peak AoI is derived, based on explicit expressions of the optimal peak AoI and the corresponding optimal system parameters including the packet arrival rate and the channel access probability are further derived. It is shown that with a given packet arrival rate (resp. a given channel access probability), the optimal channel access probability (resp. the optimal packet arrival rate), is equal to one under a small node deployment density, and decrease monotonically as the spatial deployment density increases due to the severe interference caused by spatio-temproal coupling between transmitters. When joint tuning of the packet arrival rate and channel access probability is performed, the optimal channel access probability is always set to be one. Moreover, with the sole tuning of the channel access probability, it is found that the optimal peak AoI performance can be improved with a smaller packet arrival rate only when the node deployment density is high, which is contrast to the case of the sole tuning of the packet arrival rate, where a higher channel access probability always leads to better optimal peak AoI regardless of the node deployment density. In all the cases of optimal tuning of system parameters, the optimal peak AoI linearly grows with the node deployment density as opposed to an exponential growth with fixed system parameters, which sheds important light on freshness-aware design for large-scale network.
		
	\end{abstract}
	
	\section{Introduction}
	The Internet of Things (IoT) is expected to make our physical surroundings accessible by placing sensors on everything in the world and converting the physical information into a digital format. The applications of IoT span numerous verticals, including transportation, environmental detection, and energy scheduling.
	In such applications, timely message delivery is of necessity. For example, for intelligent vehicles, real-time updates of road information are crucial for safe driving, and in environmental detection, updating environmental information on time is beneficial to the prediction, as well as preparation, for occurrence of natural disasters. Since an outdated message would become useless, timeliness is one of the critical objectives in the IoT network.
	
	To assess the timeliness of delivered messages, a new metric called the \emph{Age of Information} (AoI) was proposed in \cite{6195689}\cite{6284003}, which is defined as the time elapsed since the most recently received update was generated at its source. Aiming to design systems that can provide fresh information, extensive studies have been conducted to characterize the AoI on the basis of queuing theory. For instance, in \cite{6195689}, the AoI was minimized for first-come-first-served (FCFS) M/M/1, M/D/1, and D/M/1 queues. Multiple sources was considered in \cite{6284003} for the FCFS M/M/1 queue. In \cite{7415972}, the AoI in M/M/1/1 queue and M/M/1/2 queue with both FCFS and LCFS was characterized by considering a finite buffer capacity. The effects of the buffer size, packet age deadline, and replacement on the average AoI were further studied in \cite{7795343}.
	
	However, these studies only focused on a point-to-point communication scenario. In practice, IoT networks generally consist of a large number of nodes that intend to communicate with their destinations via spectrum, which usually constitutes a multiple access network. Due to the broadcast nature of wireless medium, transmissions of nodes will affect each other via the interference they generated. And the characterization of AoI under such a setting has attracted a variety of studies recently [5]-[14]. Specifically, the AoI was minimized in \cite{8943134}\cite{7492912} by scheduling a group of links that are active at the same time and limiting the interference to an acceptable level. Various scheduling policies including Greedy policy, stationary randomized policy, Max-Weight policy, and Whittle's Index policy were proposed in \cite{8514816} to minimize the AoI for periodic packet arrivals. In addition, the AoI performance under stationary randomized policy and Max-Weight policy was analyzed for Bernoulli packet arrivals \cite{8933047}. A joint design of status sampling and updating process to minimize the average AoI was further proposed in \cite{8778671}. Despite the promise of improving the AoI performance, the overhead of centralized scheduling may be too hefty to be affordable for IoT networks with massive connectivity. In that respect, decentralized schemes were also studied from the perspective of AoI optimization. In particular, the effectiveness of slotted ALOHA on minimizing AoI was studied in \cite{8006544}, where each node initializes a channel access attempt at each time slot with a certain probability. A threshold-based age-dependent random access protocol was proposed in \cite{9162973} \cite{yavascan2020analysis}, where each node accesses the channel only when its instantaneous AoI exceeds a predetermined threshold. A distributed transmission policy was proposed in \cite{9174254} based on the age gain which is the reduction of instantaneous AoI when packets are successfully delivered. An Index-Prioritized Random Access scheme was proposed in \cite{8935400}, where nodes access the radio channel according to their indices that reflect the urgency of update. The classic collision model was adopted in these studies, where one node can successfully access the channel if and only if there are no other concurrent transmissions. Albeit significant advances have been achieved, these works did not take into account the key effects of nodes physical attributes in wireless systems such as fading, path loss, and interference.
	
	Stochastic geometry on the other hand provides an elegant way of capturing macroscopic properties of such networks by averaging over all potential geographical patterns of the nodes, which can help to account for sources of uncertainties such as co-channel interference and channel fading. Therefore, this tool has been widely adopted to evaluate the performance of various types of wireless networks \cite{1} \cite{20160000}. Recently, there have been studies of the AoI performance in large-scale networks by combining queuing theory and stochastic geometry [17]-[23]. In particular,
	the lower and upper bounds of the average AoI for Poission bipolar network were characterized in \cite{001} via the introduction of two auxiliary systems. Based on a dominant system where every transmitter sends out packets in every time slot, \cite{002} devised a locally adaptive channel access scheme for reducing the peak AoI. In these studies, the interference was decoupled from the queue status, i.e., whether the queue is empty or not. To characterize the spatio-temporal interactions of queues, a framework was provided in \cite{003} that captures the peak AoI for large-scale IoT networks with time-triggered (TT) and event-triggered (ET) traffic. The effects of network parameters were further studied in \cite{004} \cite{9316915} on the AoI performance in the context of random access networks.
	The spatial moments of the mean AoI of the status update links were characterized in \cite{mankar2020throughput} \cite{mankar2020spatial} based on the moments of the conditional successful probability. These studies focused on a static network topology, i.e., the point process pattern is realized at the beginning of time and keeps unchanged after that, leaving the network scenario with mobility largely unexplored. In this paper, we study the optimization of AoI over a large-scale random access network with mobility. Interestingly, the expression for AoI has concise form in this case and it allows us to obtain optimal system design parameters in a close-form.
	
	In particular, we consider a Poisson bipolar network where each transmitter updates information packets according to an independent Bernoulli process. Similar to that in \cite{7415972}\cite{004}, we adopt a unit-size buffer at the transmitter side which avoids the long waiting time caused by the accumulation of data packets in the buffer. To reduce the overhead of centralized scheduling, each transmitter employs an ALOHA random access protocol, i.e., each transmitter accesses the channel with a certain probability at each time slot. The successful transmission of packets depends on the Signal to Interference plus Noise Ratio (SINR) value at the receiver side. Because of the interference, the buffer states of the transmitters are coupled with each other. By leveraging tools from stochastic geometry and queuing theory, we derive a fixed-point equation of the probability of successful transmission of each transmitter by taking into account the coupling effect. Based on the probability of successful transmission, an analytical expression for the peak AoI is obtained, which is a function of the packet arrival rate and the channel access probability. Using this expression, we find that when the node deployment density is small, the AoI performance can be always improved by choosing a large packet arrival rate or channel access probability. When the node deployment density becomes large, a very high packet arrival rate or channel access probability can in turn deteriorate the AoI performance owing to the severe interference caused by the simultaneous node transmissions. The peak AoI is then optimized by tuning the channel access probability for a given packet arrival rate and by tuning the packet arrival rate for a given channel access probability, respectively. It is found that when the packet arrival rate is optimally tuned, a higher channel access probability always leads to better peak AoI performance, but when the channel access probability is optimally tuned, the peak AoI can be benefited with a smaller packet arrival rate only when the node deployment density is high. We then study how to minimize the peak AoI by jointly tuning the packet arrival rate and the channel access probability, and find that the optimal channel access probability is always set to be one. This indicates that to reduce the waiting time in each transmitters' buffer, each packet should be transmitted as soon as possible. Yet, the packet arrival rate, i.e., the information update frequency, should be lower so as to alleviate the channel contention. For all the three cases, i.e., tuning the channel access probability, tuning the packet arrival rate and joint tuning, the optimal peak AoI grows linearly as the node deployment density increases, which is in sharp contrast to an exponential growth when the system parameters are not properly tuned. This sheds important light on freshness-aware design for large-scale networks.
	
	The remainder of this paper is organized as follows. Section \ref{system model} presents the system model and preliminary analysis. Section \ref{section:p} shows the derivation and anlysis of the probability of successful transmission. In Section \ref{PAoI}, the peak AoI is derived and optimized by tuning system parameters including the channel access probability and the packet arrival rate. Section \ref{Simulation Results} presents the simulation results of above analysis. Finally, Section \ref{conclusion} summarizes the work and draws final conclusion.
	
	\section{System Model and Preliminary analysis}\label{system model}
	Let us consider a Poisson bipolar network where transmitters are scattered according to a homogeneous Poisson point process (PPP) of density $\lambda$. As Fig. \ref{system_model} illustrates, each transmitter is paired with a receiver that is situated in distance $R$ and oriented at a random direction. In this network, the time is slotted into equal-length intervals and the transmission of each packet lasts for one slot.
	The packets arrive at each transmitter following independent Bernoulli processes of rate $\xi$. We assume every transmitter is equipped with a unit-size buffer and hence a newly incoming packet will be dropped if an elder packet is in its service. At the beginning of each time slot, transmitters with non-empty buffers will access the channel with a fixed probability $q$. To better illustrate the channel access process of each transmitter, let us define two parameters $\epsilon^{'}$ and $\epsilon^{''}$, where $\epsilon^{''}\ll\epsilon^{'}\ll1$:
	1) at $t+\epsilon^{''}$, a new packet arrives with probability $\xi$;
	2) at $t+\epsilon^{'}$, each transmitter that has one packet in its buffer accesses the radio channel with probability $q$;
	3) If the transmission is successful, then the packet departs at $t+1-\epsilon^{'}$; otherwise, the packet remains in the queue and will be sent out in the next time slot until success.%
	
	\subsection{Signal-to-Interference-plus-Noise Ratio }
	In this paper, we consider the radio frequency is globally reused, i.e., all the nodes utilize the same spectrum for packet delivery. Moreover, each transmitter employs a universally unified transmit power and thus has an equal mean received SNR $\gamma$ at the receiver. As such, for a generic transmitter $i$, its received SINR at time slot $t$ is given by
	\begin{equation}
	\text{SINR}_i(t)=\frac{h_{ii}(t)R^{-\alpha}}{\sum_{j\neq i}h_{ij}(t)e_j(t)\mathbf{1}(Q_j(t)>0)d_{ij}(t)^{-\alpha}+\gamma^{-1}} ~,
	\label{eq:defineSINR}
	\end{equation}
	where $h_{ij}(t)$ represents the small-scale fading coefficient between transmitter $j$ and receiver $i$, which is assumed to be exponentially distributed with unit mean and varies i.i.d. across space and time, $d_{ij}$ is the distance between transmitter $j$ and destination $i$, and $e_j(t)$ is a binary function, where $e_j(t)=1$ denotes that transmitter $j$ initiates the packet transmission at slot $t$, and $e_j(t)=0$ otherwise.
	$Q_j(t)$ is the queue length of transmitter $j$ at slot $t$. Thus, $\mathbf{1}(Q_j(t)>0)=1$ indicates that transmitter j has a non-empty buffer at time slot $t$, and $\mathbf{1}(Q_j(t)>0)=0 $ otherwise. The parameter $\alpha$ is the path-loss exponent. In this work, we consider a packet is successfully delivered if the received SINR exceeds a decoding threshold $\theta$. Therefore, the corresponding probability of successful transmission for node $i$ can be written as:
	\begin{equation}
	p_{i}(t)=P(\text{SINR}_i (t)>\theta).
	\label{define_p}
	\end{equation}
	\begin{figure}[t]
		\centering
		\includegraphics[width=14cm,height=6cm]{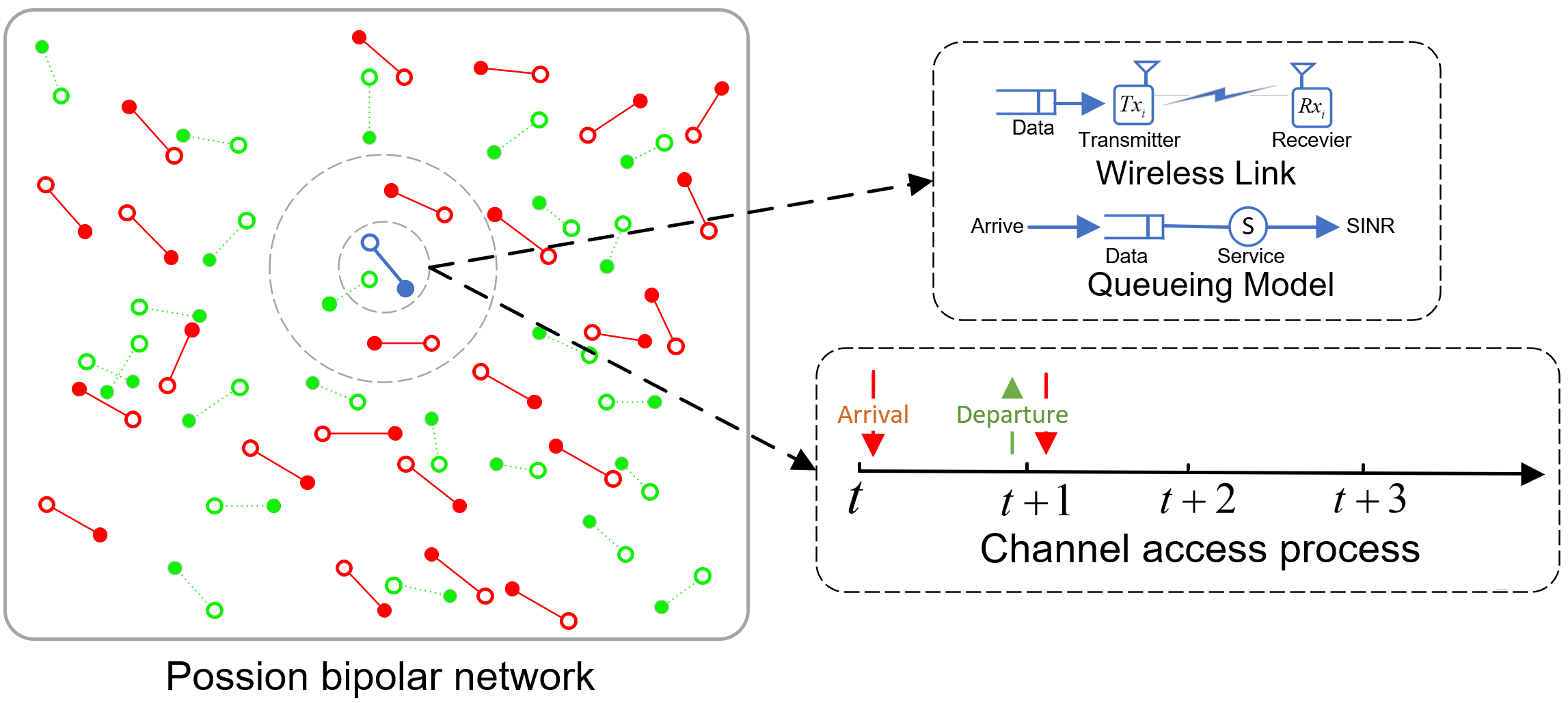}
		\caption{Snapshot of Poisson bipolar network in consideration. The up-right figure illustrates the queueing model of a generic transmiter. The down-right figure illustrates the channel access process.}
		\label{system_model}
	\end{figure}
	Similar to \cite{5601963}, we assume a high mobility random walk model for the positions of transmitters. As such, the received $\text{SINR}_i(t)$ of each transmitter $i$, $i \in \mathbb{N}$, can be considered as i.i.d across time $t$. By symmetry, the probability of successful transmission is also identical across all the transmitters. To that end, we drop the indices $i$ and $t$ in \eqref{define_p} and denote $p$ as the probability of successful transmission. Then, the dynamics of packet transmissions over each wireless link can be regarded as a Geo/Geo/1/1 queue with the service rate $qp$.
	
	\subsection{Performance Metric}
	In this paper, we focus on the performance metric of AoI, which captures the timeliness of information delivered at the receiver side. In Fig. \ref{AoIcurve}, we depict the evolution of AoI $A(t)$ over time for a Geo/Geo/1/1 queue, where $t_k$ denotes the time slot in which the $k^{th}$ packet arrived, $t^{'}_k$ denotes the time slot in which the $k^{th}$ packet is successfully transmitted, and $t^{*}_k$ denotes the time slot in which the $k^{th}$ packet is dropped. From this figure, we can see that the AoI $A(t)$ increases linearly over time and plummets at time slots $t^{'}_1, t^{'}_2, t^{'}_3,\ldots ,t^{'}_n$ where packets are successfully transmitted. Notably, during the period between $t_2$ and $t^{'}_2$, there is a packet arrivals at slot $t^\ast$ but is immediately discarded because the buffer can accommodate only one packet. Formally, the progress of such a process can be written as:
	\begin{equation}\label{define_aoi}
	A(t+1)=\left\{
	\begin{array}{lr}
	A(t)+1~~~~\text{transmission failure} \\
	t-t_k+1~~~\text{transmission successful}.
	\end{array}
	\right.\end{equation}	
	\begin{figure}[t]
		\centering
		\includegraphics[width=12cm,height=6cm]{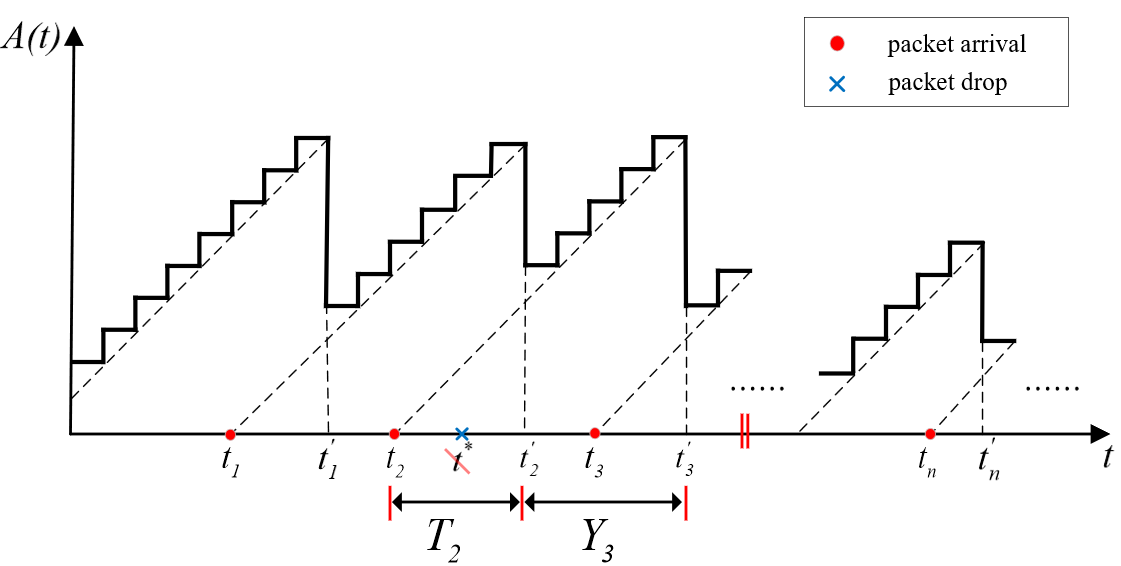}
		\caption{An example of the AoI evolution over time.}
		\label{AoIcurve}
	\end{figure}
	In this paper, we focus on the peak AoI, denoted as $A_p$, which is defined as the time average of age values at time instants when there is packet transmitted successfully as our performance metric. Such a metric is given by \cite{apdefine}
	\begin{equation}
	A_p=\lim_{T\to\infty}\sup\frac{\sum_{t=1}^{t=T}A(t)\mathbf{1}\{A(t+1)\leq A(t)\}}{\sum_{t=1}^{t=T}\mathbf{1}\{A(t+1)\leq A(t)\}}.
	\end{equation}

	\section{probability of successful transmission}\label{section:p}
	The AoI performance is dependent on the probability of successful transmission of each update packet. This section is then devoted to the characterization of the probability of successful transmission. First of all, the following lemma shows that the probability of successful transmission can be written in the form of a fixed-point equation.
	\begin{lemma}\label{lemma:p}
		The probability of successful transmission of a generic transmitter can be obtained as
		\begin{equation}
		p=\exp{\left\{-\lambda cR^2\frac{q \xi}{\xi+pq(1-\xi)}-\theta R^{\alpha}\gamma^{-1}\right\}} \label{eq:p}
		\end{equation}
		where $c=\pi\theta^{\frac{2}{\alpha}}\operatorname{sinc}(\frac{2}{\alpha})$.
		\begin{proof}
			See Appendix \ref{prooflemma1}
		\end{proof}
	\end{lemma}
	
	Lemma \ref{lemma:p} indicates that the probability of successful transmission $p$ is determined by the channel access probability $q$, the packet arrival rate $\xi$, the node deployment density $\lambda$, and the TX-RX distance $R$. The following result further characterizes the distribution of roots of \eqref{eq:p}.
	\begin{theorem}\label{Theorem_p_root}
		The fixed-point equation \eqref{eq:p} has three non-zero roots $0<p_A\leq p_S \leq p_L<1$ if $\frac{4}{q}<\lambda cR^2<\frac{((1-\xi)q+\xi)^2}{q^2\xi(1-\xi)}$ and $\xi_l<\xi<\xi_h$, where $\xi_l$ and $\xi_h$ are respectively given as follows:
		\begin{equation}
		\xi_l=\frac{q}{q+\frac{\frac{q\lambda cR^2}{2}-1-q\lambda cR^2\sqrt{\frac{1}{4}-\frac{1}{q\lambda cR^2}}}{\exp{\left\{-\theta R^\alpha\gamma^{-1}-\frac{1}{\frac{1}{2}-\sqrt{\frac{1}{4}-\frac{1}{q\lambda cR^2}}}\right\}}}}
		\label{arrival_low}
		\end{equation}
		and
		\begin{equation}
		\xi_h=\frac{q}{q+\frac{\frac{q\lambda cR^2}{2}-1+q\lambda cR^2\sqrt{\frac{1}{4}-\frac{1}{q\lambda cR^2}}}{\exp{\left\{-\theta R^\alpha\gamma^{-1}-\frac{1}{\frac{1}{2}+\sqrt{\frac{1}{4}-\frac{1}{q\lambda cR^2}}}\right\}}}};
		\label{arrival_high}
		\end{equation}
		otherwise, \eqref{eq:p} has only one non-zero root $0<p_L\leq 1$.
	\end{theorem}
	\begin{proof}
		See Appendix \ref{studyProot}
	\end{proof}
	According to the approximate trajectory analysis \cite{6205590}, not all roots in \eqref{eq:p} are steady-state points. More precisely, we have the following:\\
	1) if \eqref{eq:p} has only one non-zero root $p_L$, then $p_L$ is a steady-state point. \\
	2) if \eqref{eq:p} has three non-zero roots $0<p_A\leq p_S \leq p_L\leq1$, only $p_L$ and $p_A$ are steady-state point.
	
	Corollary 1 further summarises the properties of the steady-state points with regard to the system parameters.
	
	\begin{corollary}
		The steady-state points $p_A$ and $p_L$ are monotonic decreasing functions of the node deployment density $\lambda$, the packet arrival rate $\xi$, and the transmission probability $q$,
	\end{corollary}
	\begin{proof}
		See Appendix \ref{studyPLtrend}
	\end{proof}

	\section{Peak Age of Information}\label{PAoI}
	In this section, we first derive an expression of the peak AoI in the employed random-access Poisson network. Based on that, we then optimize the peak AoI by tuning the channel access probability $q$ and the packet arrival rate $\xi$.
	
	According to Fig. \ref{AoIcurve}, the $k^{th}$ packet's service time can be expressed as $T_k=t^{'}_k-t_k$, and the inter-departure time between the $(k-1)^{th}$ packet and $k^{th}$ packet can be written as $Y_k=t^{'}_k-t^{'}_{k-1}$. Then, the peak AoI can be expressed as $A_p=E[T_k]+E[Y_k]$. Lemma \ref{PAoIdef} gives an explicit expression for the peak AoI.
	\begin{lemma}\label{PAoIdef}
		The peak AoI $A_p$ can be written as:
		\begin{equation}\label{eq:PeakAge}
		A_p=E[T_{k-1}]+E[Y_k]=\frac{1}{\xi}+\frac{2}{qp}-1.
		\end{equation}
		\begin{proof}
			See Appendix \ref{proofPAoI}.
		\end{proof}
	\end{lemma}

	Lemma \ref{PAoIdef} shows that the peak AoI $A_p$ is affected by the channel access probability $q$ and the packet arrival rate $\xi$. Consequently, it is of great importance to explore how to properly tune the channel access probability $q$ and the packet arrival rate $\xi$ so as to minimize the peak AoI $A_p$, we then establish the following optimization problem
	\begin{equation}\label{eq:optimization}
	\begin{aligned}
	A_p^*= & \underset{\{   q, \xi\} }{\text{min}}   \quad A_p \\
	& \text{s.t. } \quad    q\in (0,1],\\
	&  \quad  \quad  ~ \xi \in (0,1].
	\end{aligned}
	\end{equation}
	
	In the following, by assuming $\xi$ is given, we start by the optimal tuning of the channel access probability $q$, and then by assuming $q$ is given, we focus on the optimal tuning of the packet arrival rate $\xi$. Finally, we study how to jointly tune $q$ and $\xi$ to address the optimization problem in \eqref{eq:optimization}.

	\subsection{Optimal Tuning of Channel Access Probability $q$}
	The following theorem presents the optimal channel access probability $q^\ast_\xi$ that minimizes the peak AoI $A_p$, i.e., $A^{q=q^\ast_\xi}_p=\underset{q}{\min}A_p$.
	
	\begin{theorem}\label{Theorem_OPtimalQ}
		Given a packet arrival rate $\xi$, the optimal peak AoI $A^{q=q^{*}_{\xi}}_p$ is given by
		\begin{equation}\label{eq:OptimalQ2}
		A_{p}^{q=q^{*}_{\xi}} =  \begin{cases}
		2\lambda c R^2\exp{\left\{\theta R^\alpha \gamma^{-1}+1\right\}}-\frac{1}{\xi}+1
		\quad  &\text{if } \lambda c R^2>1+\frac{{p}_{*}(1-\xi)}{\xi} \\
		\frac{1}{\xi}+\frac{2}{p^{*}}-1 \quad &\text{otherwise},\end{cases}
		\end{equation}
		which is acheived when the channel access probability $q$ is set to be
		\begin{equation}\label{eq:OptimalQ}%
		q=q^\ast_\xi =  \begin{cases}
		\frac{1}{\lambda cR^2-\frac{1-\xi}{\xi}\exp{\left\{-\theta R^\alpha \gamma^{-1}-1\right\}}} \quad  &\text{if } \lambda c R^2>1+\frac{{p}_{*}(1-\xi)}{\xi} \\
		1 \quad &\text{otherwise},
		\end{cases}
		\end{equation}
		where $p_{*}$ is the non-zero root of the following equation
		\begin{equation}
		p_{*}=\exp\left\{-\lambda c R^2 \frac{\xi}{\xi+p_{*}(1-\xi)}-\theta R^\alpha \gamma^{-1}\right\}.\label{eq:q1valuep}
		\end{equation}
		
	\end{theorem}
	\begin{proof}
		See Appendix \ref{prooftheorem2}
	\end{proof}
	Theorem \ref{Theorem_OPtimalQ} shows that the optimal channel access probability $q^\ast_\xi=1$ when $\lambda c R^2\leq 1+\frac{{p}_{*}(1-\xi)}{\xi}$, indicating that in this case, each node would transmit its packet as long as the buffer is nonempty. As the node deployment density $\lambda$, the distance between each TX-RX distance $R$ or the decoding threshold $\theta$ (equivalently, $c$ according to \eqref{eq:p}) grows, we have $q^\ast_\xi<1$ due to either mounting channel contention or a lower chance of successful packet decoding.%

	To take a closer look at Theorem \ref{Theorem_OPtimalQ}, Fig. \ref{fig:OptimalQandPage} demonstrates how the optimal channel access probability $q^\ast_\xi$ and the corresponding the peak AoI $A^{q=q^\ast_\xi}_p$ vary with the node deployment density $\lambda$ under different values of the packet arrival rate $\xi$. It can be seen that when $\lambda$ is small, e.g., $\lambda=0.02$, the optimal channel access probability $q^\ast_\xi=1$ regardless of the value of the packet arrival rate $\xi$. Yet, the peak AoI $A^{q=q^\ast_\xi}_p$ crucially depends on $\xi$. Intuitively, a smaller node deployment density can reduce the interference among transmitter-receiver pairs, which improves the probability of successful packet transmission. Accordingly, the age performance could be effectively improved with more frequent updates, i.e., a larger packet arrival rate $\xi$. Therefore, as shown in Fig. \ref{fig:OptimalQandPage}b, $A^{q=q^\ast_\xi}_p$ with $\xi=0.9$ is lower than those with $\xi=0.6$ or $\xi=0.3$ when the node deployment density $\lambda$ is small. On the other hand, if the node deployment density $\lambda$ grows, then to relieve the channel contention, the system should reduce the channel access probability. Thus, we can see $q^\ast_\xi$ decreases with $\lambda$ and the descent position, i.e., the starting point that $q^\ast_\xi<1$, is positively correlated with the packet arrival rate $\xi$. In this case, it is interesting to observe that the peak AoI $A^{q=q^\ast_\xi}_p$ can be benefited with a lower packet arrival rate $\xi$, which is in sharp contrast to the case where the node deployment density $\lambda$ is small.
	\begin{figure}[t]%
		\begin{minipage}[t]{0.5\textwidth}
			\includegraphics[width=8cm,height=6.73cm]{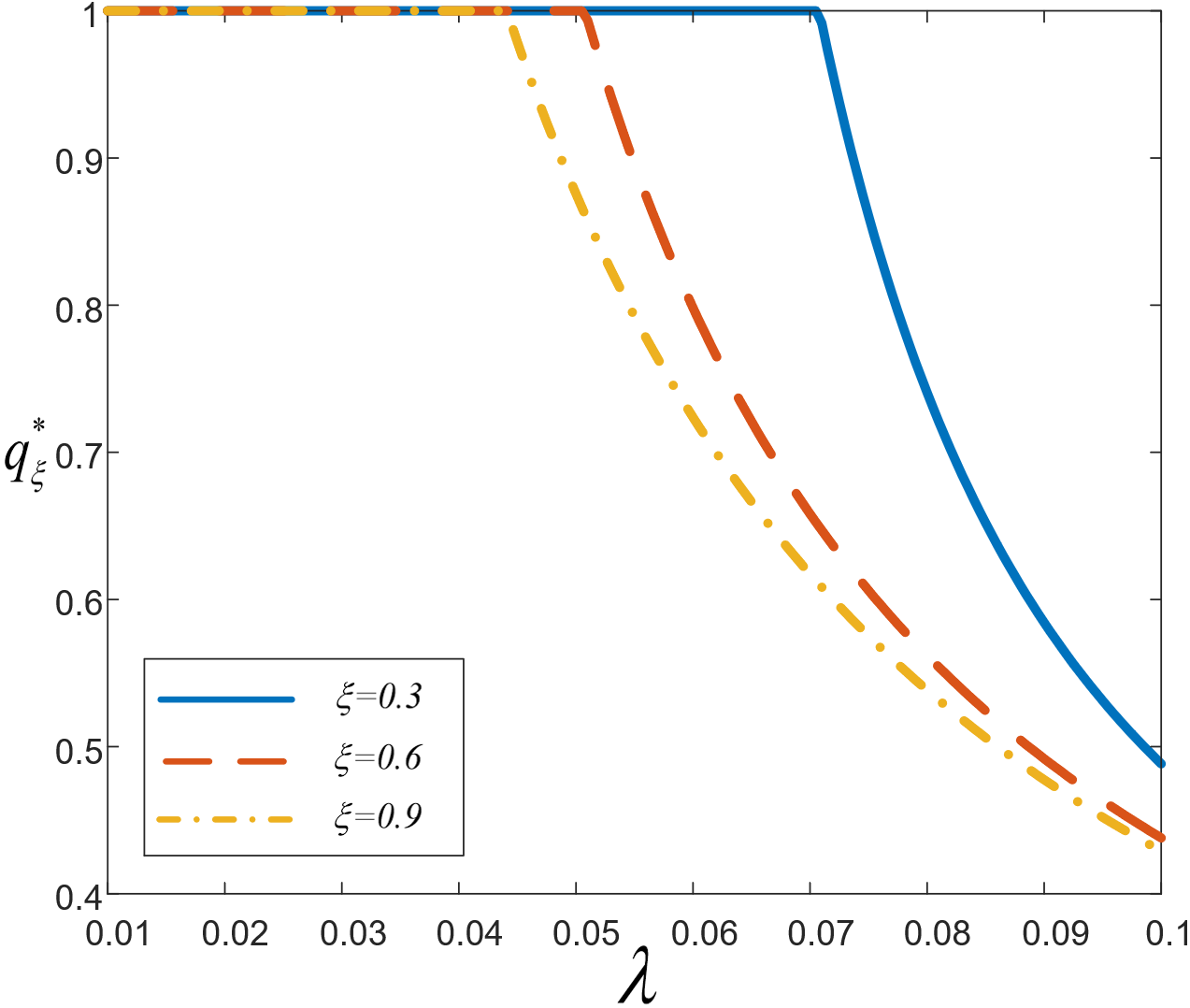}
			\centering{(a)}
			\label{fig:optimal:q:lambda}
		\end{minipage}%
		\begin{minipage}[t]{0.5\textwidth}
			\includegraphics[width=8cm,height=6.7cm]{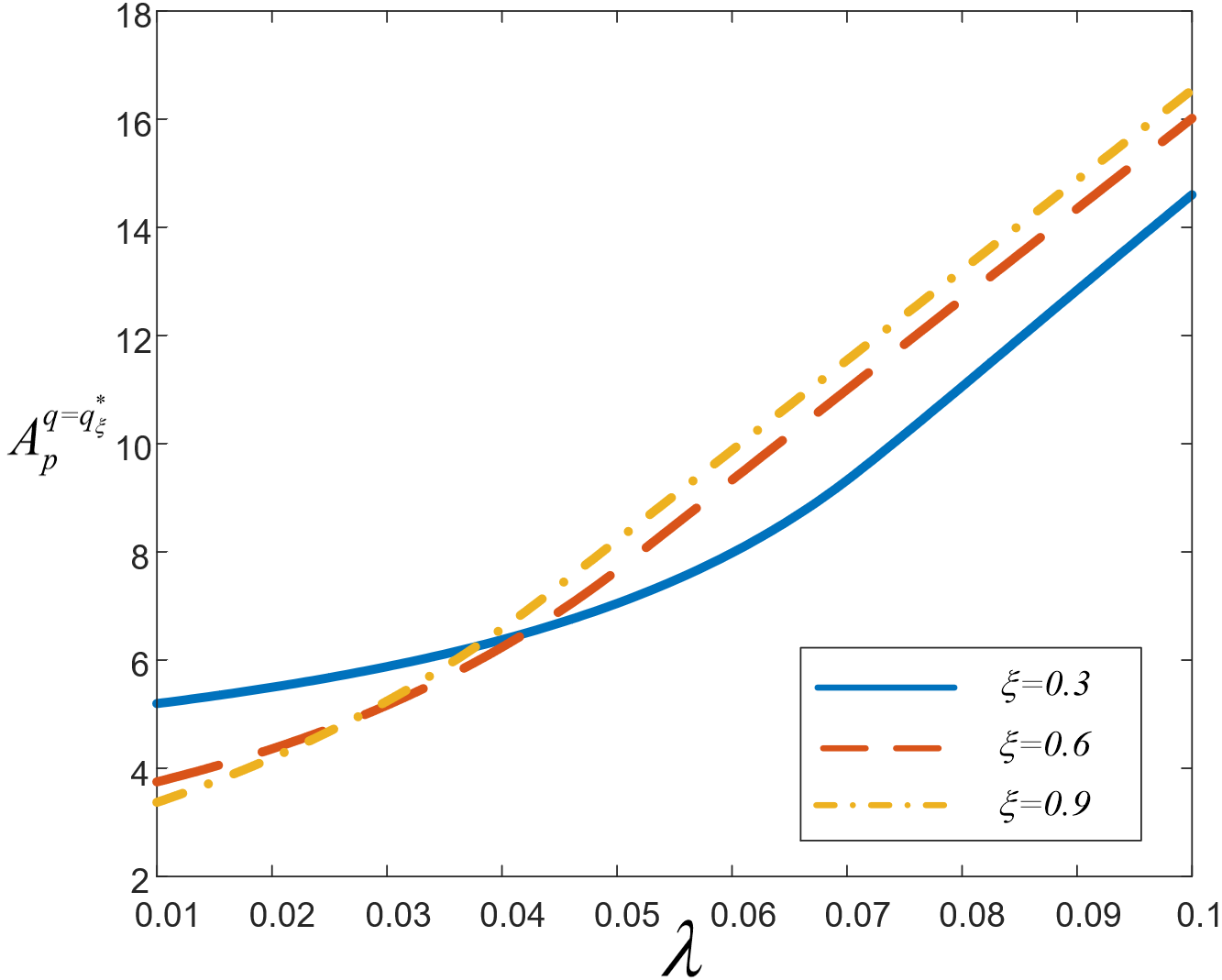}
			\centering{(b)}
			\label{fig:optimal:q:pAoi}
		\end{minipage}
		\caption{ Optimal channel access probability $q^\ast_\xi$ and the corresponding the peak AoI $A^{q=q^\ast_\xi}_p$ versus the node deployment density $\lambda$. $\alpha=3$, $\theta=0.2$, $\gamma=20$, $R=3$. $\xi\in\{0.3,0.6,0.9\}$. (a) $q^\ast_\xi$ versus $\lambda$. (b) $A^{q=q^\ast_\xi}_p$ versus $\lambda$.} \label{fig:OptimalQandPage}
	\end{figure}

	\begin{figure}[t]%
		\begin{minipage}[t]{0.5\textwidth}
			\includegraphics[width=8cm,height=6.85cm]{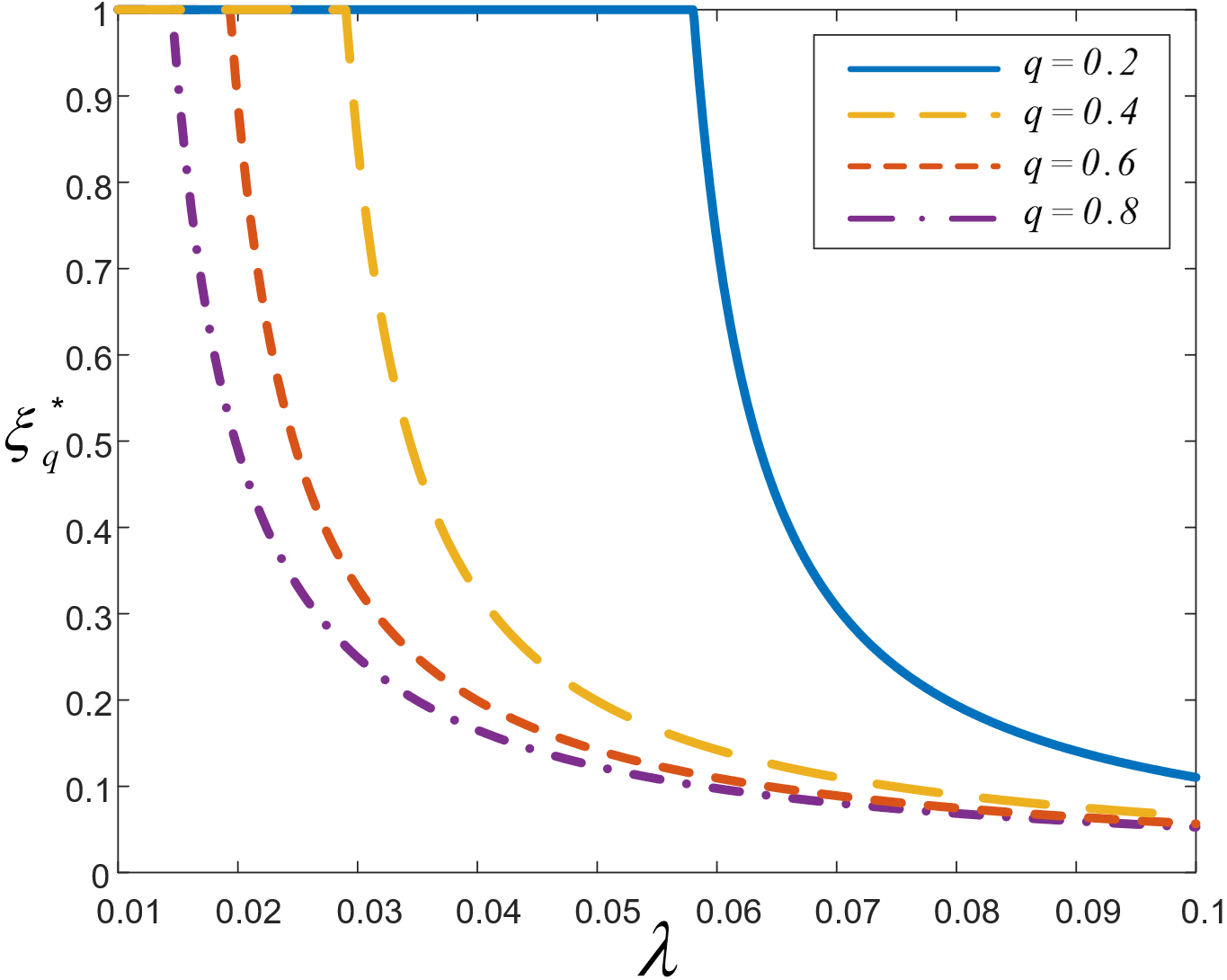}
			\label{fig:optimal:alpha:lambda}
			\centering{(a)}
		\end{minipage}%
		\begin{minipage}[t]{0.5\textwidth}
			\includegraphics[width=8.2cm,height=6.85cm]{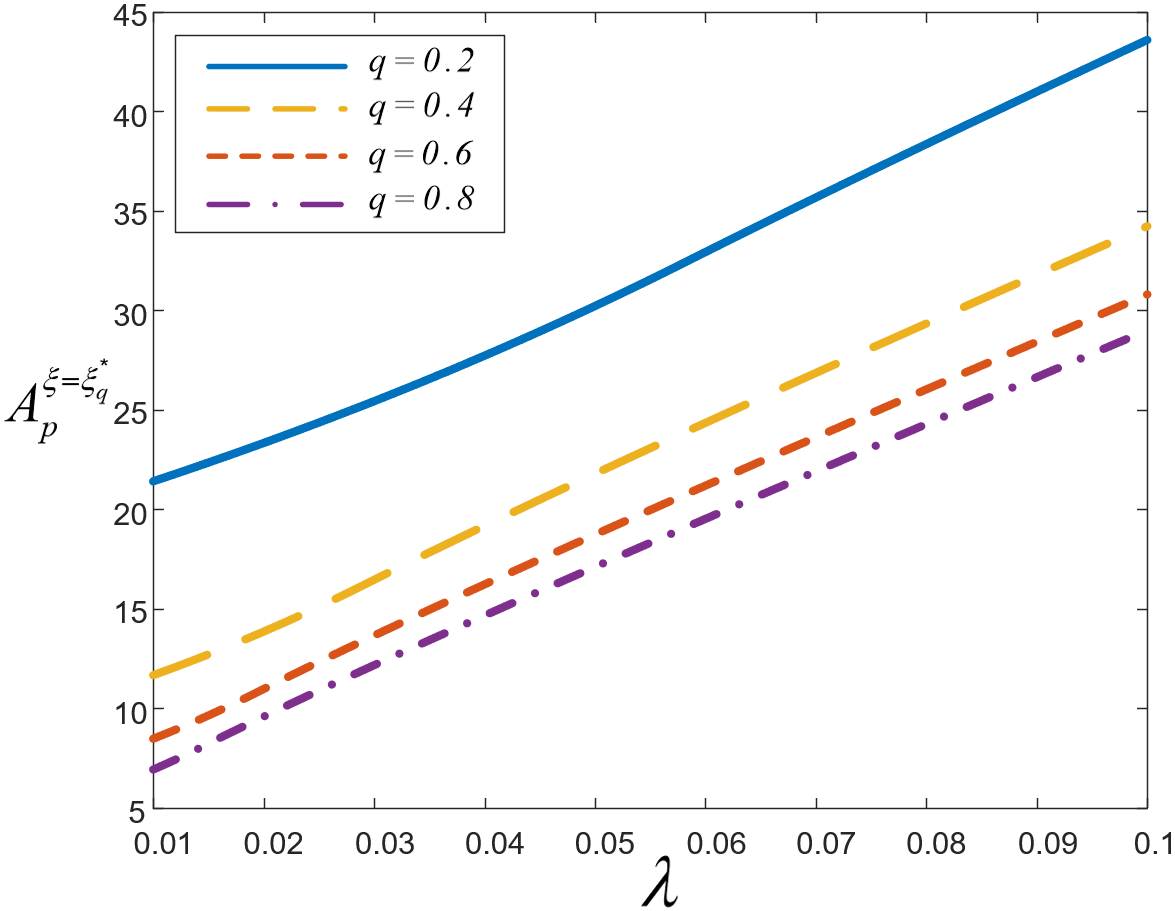}
			\label{fig:optimal:alpha:lambda:PAoI}
			\centering{(b)}
		\end{minipage}
		\caption{Optimal channel access probability $\xi^\ast_q$ and the corresponding the peak AoI $A^{\xi=\xi^\ast_q}_p$ versus the node deployment density $\lambda$. $\alpha=3$, $\theta=0.5$, $\gamma=20$, $R=3$. $q\in\{0.2,0.4,0.6,0.8\}$. (a) $\xi^\ast_q$ versus $\lambda$. (b) $A^{\xi=\xi^\ast_q}_p$ versus $\lambda$.}
		\label{fig:optimal_q_Ap}
	\end{figure}%
	\subsection{Optimal Tuning of Packet Arrival Rate $\xi$}
	The following theorem presents the optimal packet arrival rate $\xi^\ast_q$ that minimizes the peak AoI $A_p$, i.e., $A^{\xi=\xi^\ast_q}_p=\underset{\xi}{\min}A_p$.
	\begin{theorem}\label{Theorem_OPtimalAlpha}
		Given a channel access probability  $q$, the optimal peak AoI $A^{\xi=\xi^\ast_q}_p$ is given by
		\begin{equation}\label{eq:OPtimalAlpha2}
		A_{p}^{\xi=\xi^{*}_{q}} =  \begin{cases}
		\frac{q\lambda cR^2\left(\sqrt{1+\frac{4}{q\lambda cR^2}}+1\right)+2}{2q \exp{\left\{-\frac{2}{\sqrt{1+\frac{4}{q\lambda cR^2}}+1}-\theta R^\alpha \gamma^{-1}\right\}}} \quad &\text{if } \lambda c R^2>\frac{1}{2q} \\
		\frac{2}{q}\exp{\left\{\lambda cR^2q+\theta R^\alpha \gamma^{-1}\right\}} \quad &\text{otherwise},\end{cases}		
		\end{equation}
		which is acheived when the packet arrival rate $\xi$ is set to be
		\begin{equation}\label{eq:OPtimalAlpha}%
		\xi=\xi_q^{*} =  \begin{cases}
		\frac{2q \exp{\left\{-\frac{2}{\sqrt{1+\frac{4}{q\lambda cR^2}}+1}-\theta R^\alpha \gamma^{-1}\right\}}}{q\lambda cR^2\left(\sqrt{1+\frac{4}{q\lambda cR^2}}+1\right)+2q\exp{\left\{-\frac{2}{\sqrt{1+\frac{4}{q\lambda cR^2}}+1}-\theta R^\alpha \gamma^{-1}\right\}}-2}  &\text{if } \lambda c R^2>\frac{1}{2q} \\
		1 \quad &\text{otherwise}.
		\end{cases}
		\end{equation}	
	\end{theorem}
	\begin{proof}
		See Appendix \ref{prooftheorem3}
	\end{proof}%
	
	Theorem \ref{Theorem_OPtimalAlpha} reveals that the optimal packet arrival rate $\xi^\ast_q=1$ when $\lambda c R^2<\frac{1}{2 q}$, indicating that in this case, to minimize the peak AoI, new packets shall be updated as frequent as possible. Similarly to Theorem \ref{Theorem_OPtimalQ}, as $\lambda$, $R$ or $c$ grows, due to mounting channel contention or a lower probability of successful transmission, the optimal packet arrival rate $\xi^\ast_q<1$.%
	
	Fig. \ref{fig:optimal_q_Ap} demonstrates how the optimal packet arrival rate $\xi^\ast_q$ and the corresponding the peak AoI $A^{\xi=\xi^\ast_q}_p$ varies with the node deployment density $\lambda$ under various value of the channel access probability.
	Intuitively, as the node deployment density $\lambda$ increases, to reduce the interference, the system should either reduce the channel access probability $q$ or the packet arrival rate $\xi$. Accordingly, we can see from Fig. \ref{fig:optimal_q_Ap}a that as $\lambda$ increases, the optimal packet arrival rate $\xi^\ast_q$ declines, which could be further reduced with a larger channel access probability $q$. On the other hand, as shown in Fig. \ref{fig:optimal_q_Ap}b, the corresponding the peak AoI $A^{\xi=\xi^\ast_q}_p$ grows with the node deployment density $\lambda$, which is intuitively clear. Yet, in contract to that in Fig. \ref{fig:OptimalQandPage}b, a larger channel access probability $q$ always leads to a smaller $A^{\xi=\xi^\ast_q}_p$.
	\subsection{Joint Tuning of $q$ and $\xi$}
	
	The following theorem presents the result of joint tuning of the probability of successful transmission $q$ and the packet arrival rate $\xi$ for minimizing the peak AoI $A_p$, i.e., $A^{*}_p=\underset{\{q,  \xi\}}{\min} A_p.$
	
	\begin{theorem}\label{Theorem_OPtimalqAlphaboth}
		The optimal peak AoI $A^{*}_p=\underset{\{q, \xi\}}{\min}~A_p$ is given by
		\begin{equation}\label{eq:OPtimalAlpha2}
		A_{p}^{*} =  \begin{cases}\frac{\lambda cR^2\left(\sqrt{1+\frac{4}{\lambda cR^2}}+1\right)+2}{2\exp{\left\{-\frac{2}{\sqrt{1+\frac{4}{\lambda cR^2}}+1}-\theta R^\alpha \gamma^{-1}\right\}}},
		\quad &\text{if } \lambda c R^2>\frac{1}{2} \\
		2\exp{\left\{\lambda cR^2+\theta R^\alpha \gamma^{-1}\right\}} \quad &\text{otherwise},\end{cases}
		\end{equation}
		which is achieved when the channel access probability $q$ is set to be
		\begin{equation}
		q=q^{*}=1,
		\end{equation}
		and the packet arrival rate $\xi$ is set to be
		\begin{equation}\label{eq:OPtimalAlpha2}
		\xi=\xi^{*} =  \begin{cases}
		\frac{2 \exp{\left\{-\frac{2}{\sqrt{1+\frac{4}{\lambda cR^2}}+1}-\theta R^\alpha \gamma^{-1}\right\}}}{\lambda cR^2\left(\sqrt{(1+\frac{4}{\lambda cR^2}}+1\right)+2\exp{\left\{-\frac{2}{\sqrt{1+\frac{4}{\lambda cR^2}}+1}-\theta R^\alpha \gamma^{-1}\right\}}-2}  &\text{if } \lambda c R^2>\frac{1}{2} \\
		1 \quad &\text{otherwise}.
		\end{cases}
		\end{equation}
		\begin{proof}
			See Appendix \ref{prooftheorem4}
		\end{proof}%
	\end{theorem}
	
	\begin{figure}[t]
		\begin{minipage}[t]{0.5\textwidth}
			\includegraphics[width=7.8cm,height=7.10cm]{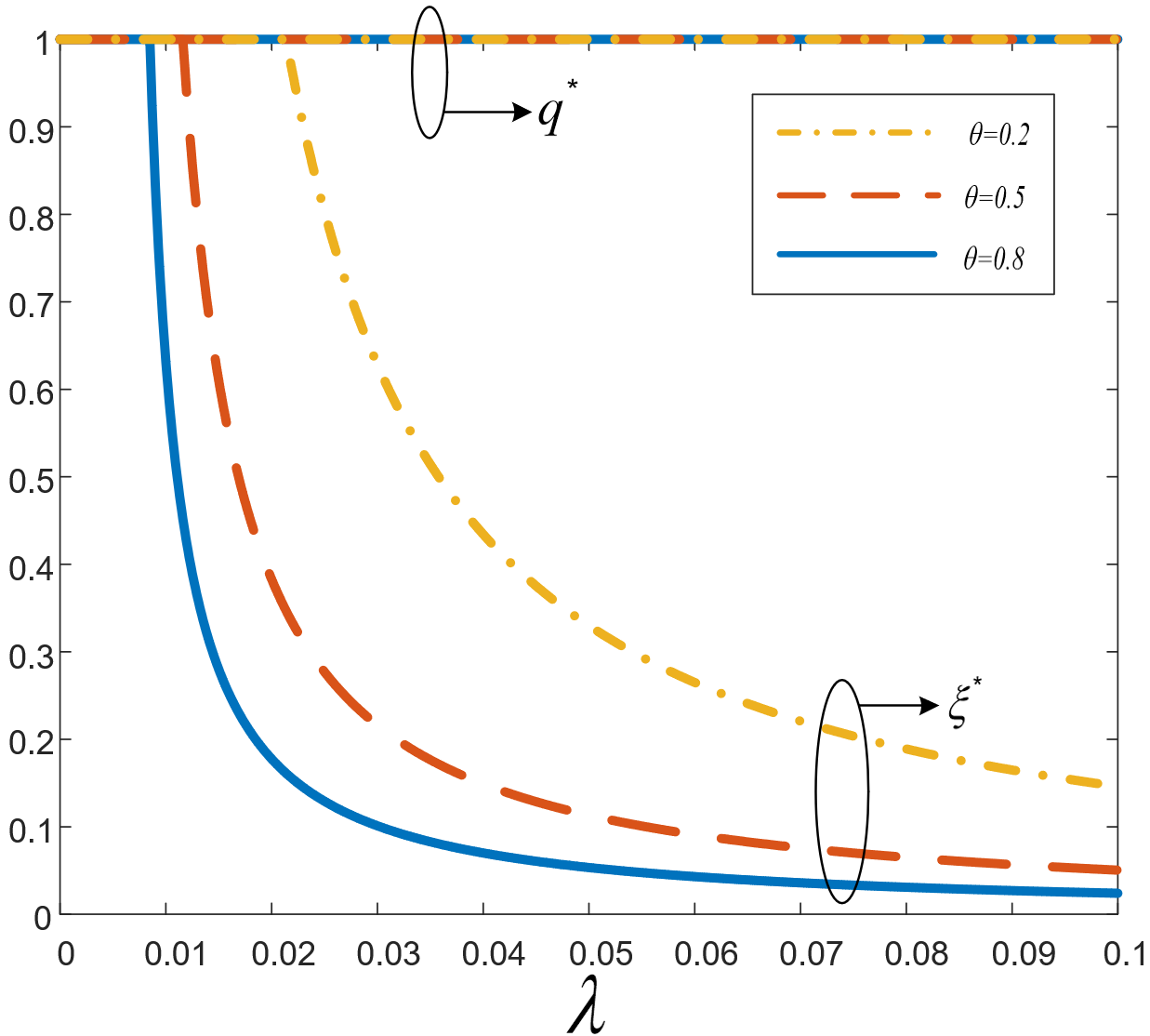}
			\label{fig:optimal:alphaq:lambda}
			\centering{(a)}
		\end{minipage}%
		\begin{minipage}[t]{0.5\textwidth}
			\includegraphics[width=8.5cm,height=7.08cm]{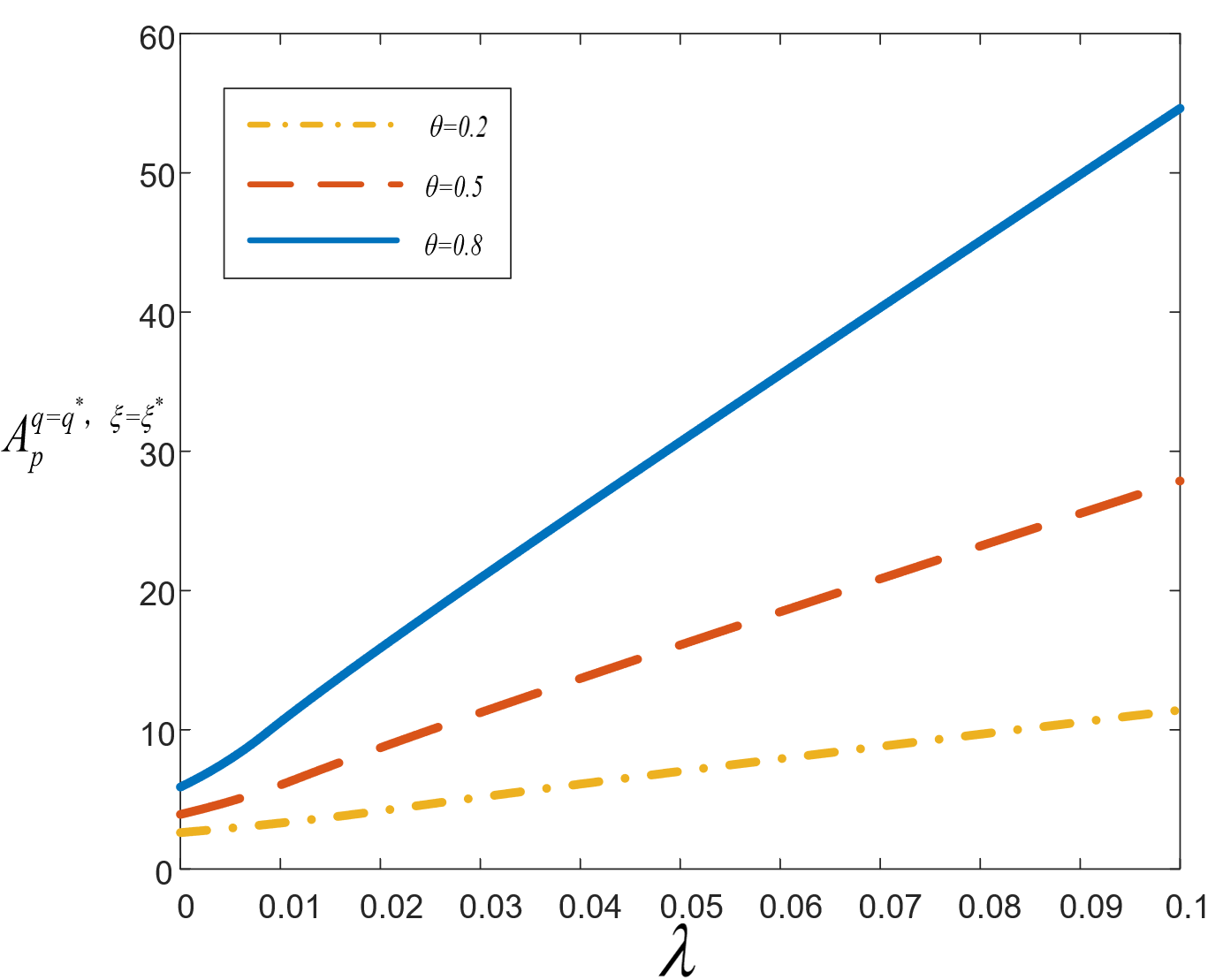}
			\label{fig:optimal:alphaq:lambda:PAoI}
			\centering{(b)}
		\end{minipage}
		\caption{Optimal the channel access probability $q^{\ast}$, the packet arrival rate $\xi^{\ast}$ and the corresponding peak AoI $A^{\xi=\xi^{\ast},q=q^{\ast}}_p$ versus the node deployment density $\lambda$. $\alpha=3$, $\gamma=20$, $R=3$, $\theta\in\{0.2,0.5,0.8\}$. (a) $\xi^{\ast}$, $q^{\ast}$ versus $\lambda$. (b) $A^{\xi=\xi^{\ast},q=q^{\ast}}_p$ versus $\lambda$.}
		\label{fig:optimal_q_alpha_both}
	\end{figure}
	In Fig. \ref{fig:optimal_q_alpha_both}, which demonstrate how the optimal channel access probability $q^\ast$, the optimal packet arrival rate $\xi^\ast$ and the corresponding minimum peak AoI $A^\ast_p$ vary with the node deployment density $\lambda$.
	It is clear from Fig. \ref{fig:optimal_q_alpha_both}  that the optimal packet arrival rate $\xi^\ast=1$ when $\lambda c R^2<\frac{1}{2}$, indicating that in this case, to minimize the peak AoI, packets should arrival at the system as many as possible. Similarly to Theorem \ref{Theorem_OPtimalQ}, as $\lambda$, $R$ or $c$ grows, due to mounting channel contention or a lower probability of successful transmission, the optimal packet arrival rate $\xi^\ast<1$. In such a jointly optimization, it is interesting to see that we always have $q^{\ast}=1$ whatever the value of $\lambda$ is. Intuitively, as $\lambda$ increases, to reduce the channel contention, the system should decrease both $q$ and $\xi$. Yet, the shorter a period the packet stays in the buffer, the lower AoI will be when it is updated. Accordingly, the system keeps the optimal channel access probability $q^\ast=1$ while reduces the packet arrival rate $\xi^\ast$ only.

	\section{Simulation Results and Discussions}\label{Simulation Results}
	In this section, we provide simulation results to validate the analysis and further shed light on AoI minimum network designs. Specifically, in the begining of each simulation run, we realize the locations of transmitter-receiver pairs over a $100\times100~m^2$ square area according to independent PPPs and place the typical link where the recevier is located at the center of the area. In each timeslot, the location of each pair is shifted except for the typical link, and each simulation last for $10^5$ time slots. In each realization, the simulated peak AoI is caculated as the sum of the peak value of AoI curve to the number of successful transmissions of the typical link. To obtain the simulated mean peak AoI, we average over 20 realizations.
	
	\begin{figure}[t]
		\centering
		\includegraphics[width=8.66cm,height=7cm]{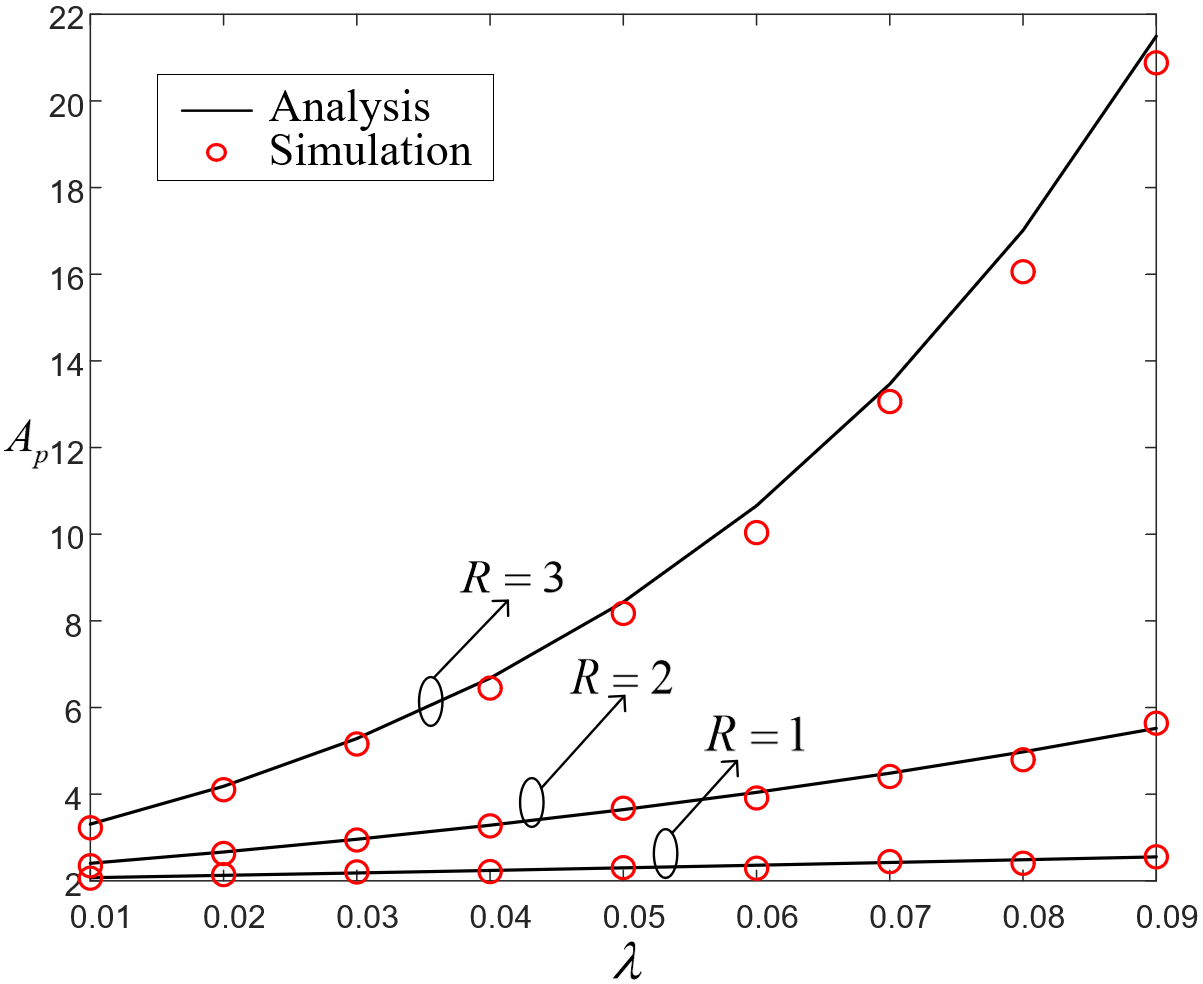}
		\caption{Peak Age of Information $A_p$ versus the node deployment density $\lambda$. $\alpha=3$, $\theta=0.2$, $q=1$, $\xi=1$, $\gamma=20$, $R\in\{1,2,3\}$.}
		\label{PAoI_density}
	\end{figure}
	
	Fig. \ref{PAoI_density} illustrates how the peak AoI $A_p$ varies with the node deployment density $\lambda$ under different TX-RX distances. From this figure, we can see that the simulation results match well with the analysis, which verify the accuracy of Theorem \ref{PAoIdef}. Moreover, following the developed analysis, we know that as the node deployment density goes up, the interference among the wireless link become more severe, leading to a smaller probability of successful transmissions that deteriorates the age performance. Accordingly, we can see from Fig. \ref{PAoI_density} that the peak AoI $A_p$ increases as
	$\lambda$ increases. By reducing the TX-RX distance of each pair the SINR can be improved. As Fig. \ref{PAoI_density} illustates, when the distance of each pair is reduced to be $R=1$, the peak AoI $A_p$ becomes less sensitive to the variation of the node deployment density $\lambda$.
	\begin{figure}[t]
		\centering
		\includegraphics[width=8.2cm,height=6.9cm]{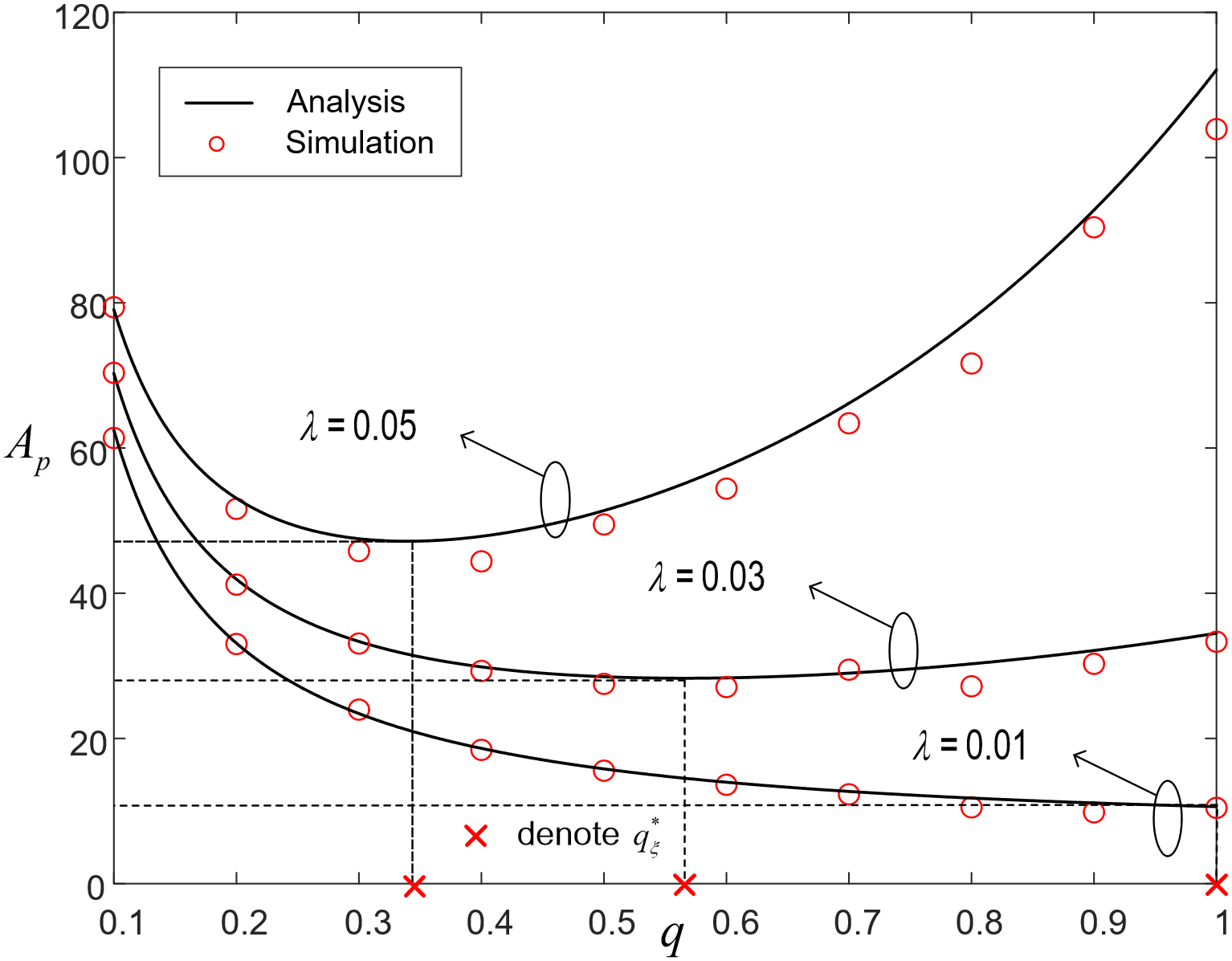}
		\caption{Peak Age of Information $A_p$ versus the channel access probability $q$. $R=3$, $\alpha=3$, $\theta=0.8$, $ \xi =1$, $\gamma=20$, $\lambda\in\{0.01,0.03,0.05\}$.}
		\label{PAoI_transmission}
	\end{figure}
	
	Fig. \ref{PAoI_transmission} depicts the peak AoI $A_p$ as a function of the channel access probability $q$ under various values of the node deployment density $\lambda$. It can be seen that when the node deployment density $\lambda$ is small, $A_p$ monotonically decreases with respect to $q$ and reaches the smallest value when $q=1$. Note that an increase in the channel access probability has two opposite effects on the peak AoI. On the one side, a larger channel access probability $q$ results in a short waiting time in the buffer which reduces the staleness of information packet. On the other side, a larger channel access probability $q$ can also incur severe interference due to a high offered loads of each transmitter's queue, leading to a lower probability of successful transmission. As Fig. \ref{PAoI_transmission} illustrates, the node deployment density $\lambda$ determines how these two effects trade off with each other. In particular, with a small node deployment density, the interference would not become severe even the channel access probability $q$ is large. With a large $\lambda$, on the other hand, the interference can deteriorate the AoI if $q$ is large.
	\begin{figure}[t]
		\centering
		\includegraphics[width=8cm,height=6.5cm]{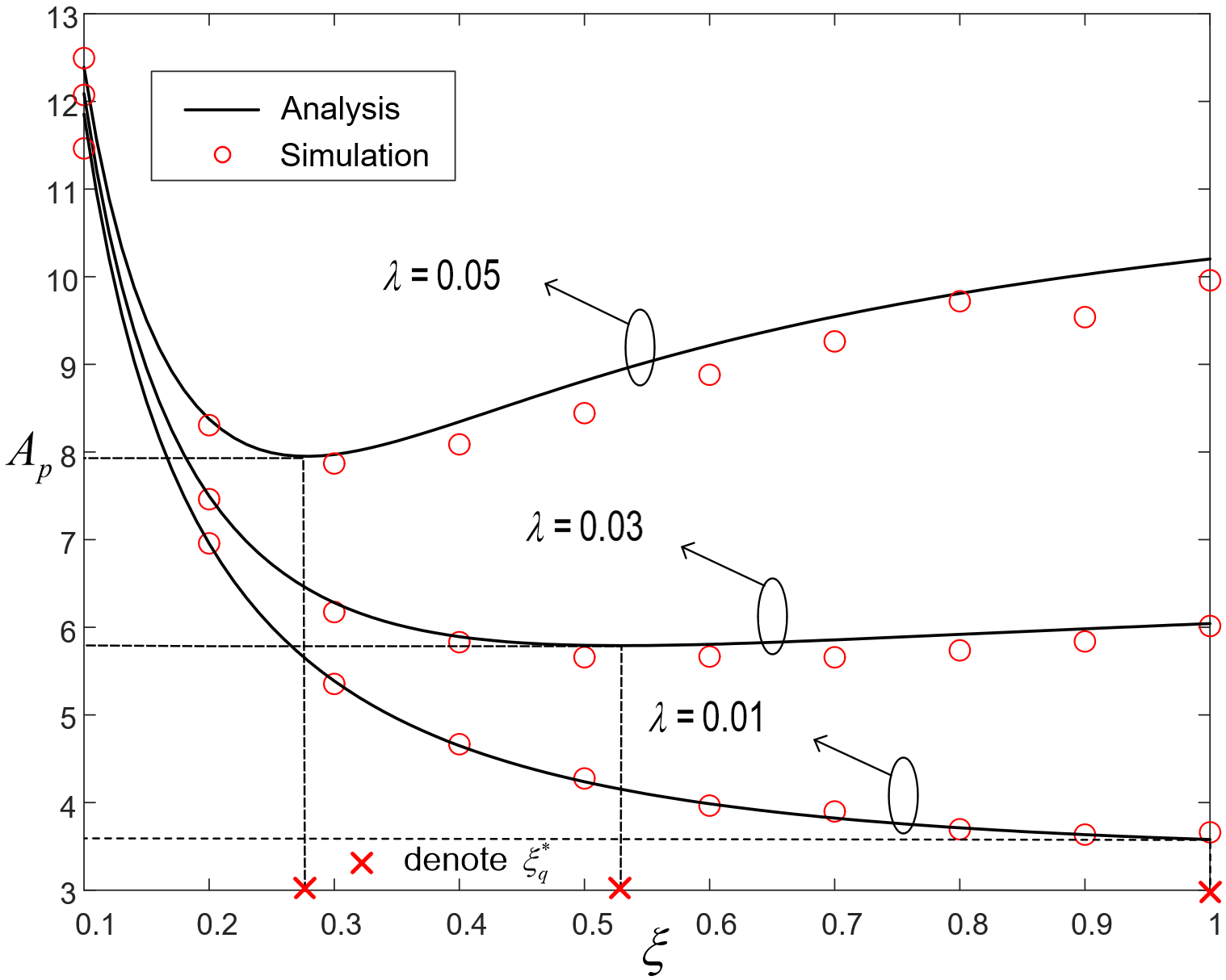}
		\caption{Peak Age of Information $A_p$ versus the packet arrival rate $\xi$. $R=2$, $\alpha=3$, $\theta=0.8$, $q=1$, $\gamma=20$, $\lambda\in\{0.01,0.03,0.05\}$.}
		\label{PAoI_arrival_rate}
	\end{figure}
	Similar observations can also be seen from Fig. \ref{PAoI_arrival_rate}, which demonstrates how the peak AoI $A_p$ varies with the packet arrival rate $\xi$ under different values of the node deployment density.
	
	The optimal channel access probability $q_{\xi}^{*}$ under a given packet arrival rate $\xi$ and the optimal packet arrival rate $\xi^{*}_{q}$ for a given channel access probability $q$ are illustrated in Fig. \ref{PAoI_transmission} and Fig. \ref{PAoI_arrival_rate}, respectively. It can be seen that by optimally tuning $q$ and $\xi$, the peak AoI can be largely reduced. To futher investigate the performance gain brought by joint tuning the channel access probability $q$ and the packet arrival rate $\xi$, Fig. \ref{fig:optimal_q_alpha} demonstrates how the peak AoI $A_p$ varies with the node deployment density $\lambda$ in four cases: 1) fixed $\xi$ and $q$, 2) optimal $q$ with fixed $\xi$, 3) optimal $\xi$ with fixed $q$, 4) joint optimal tuning of $q$ and $\xi$. We can clearly see that with fixed $\xi$ and $q$, the peak AoI $A_p$ exponentially increases with $\lambda$. In sharp contrast, with a joint optimal tuning of $q$ and $\xi$, the peak AoI $A_p$ linearly increases with $\lambda$. It implies that the performance gain becomes significant when $\lambda$ is large.
	\begin{figure}[t]
		\centering
		\includegraphics[width=8cm,height=6.8cm]{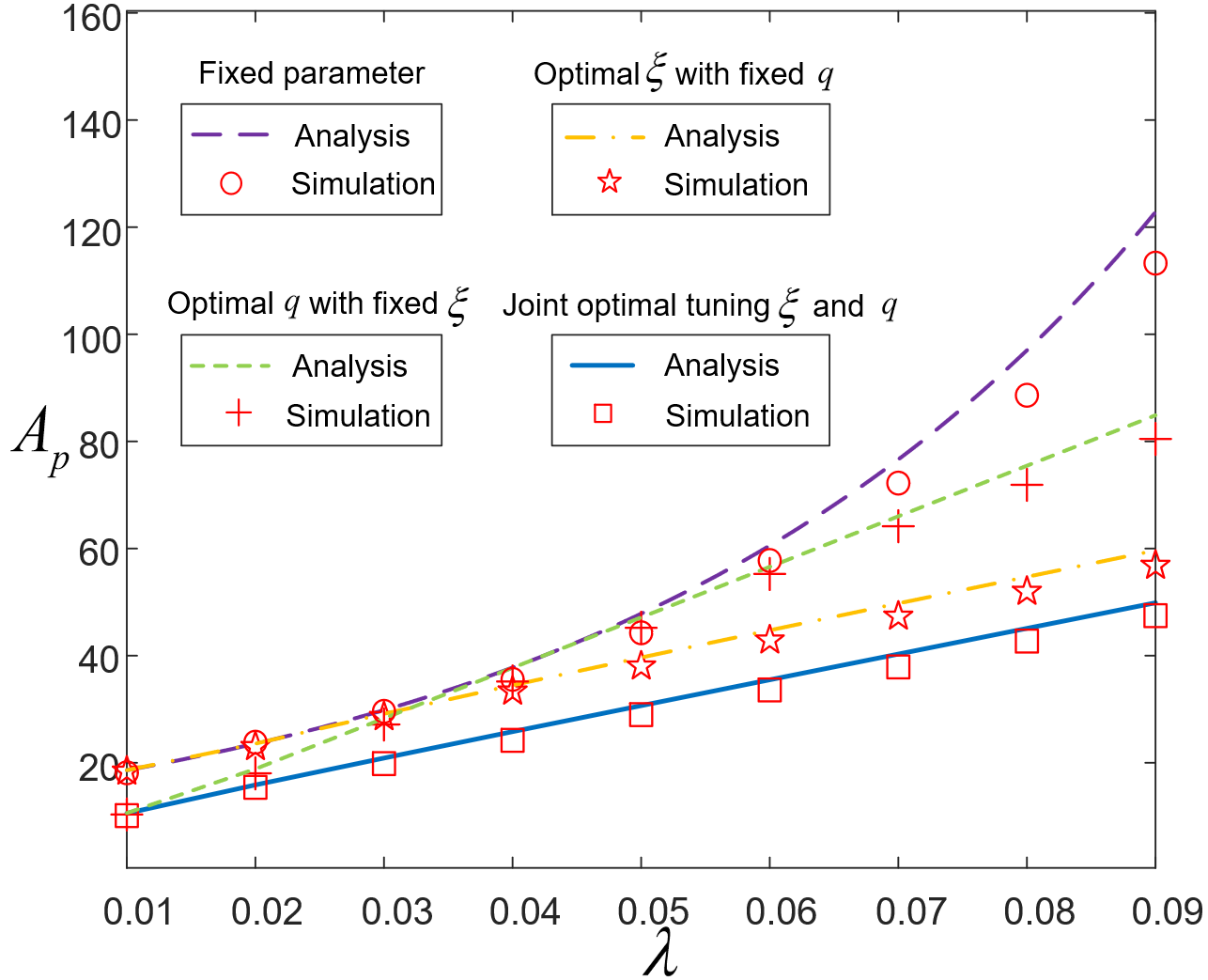}
		\caption{Optimal Peak Age of Information $A_p$ versus the node deployment density. 1) fixed parameter: $\xi=1$, $q=0.6$, 2) optimal $q$ with fixed $\xi=1$, 3) optimal $\xi$ with fixed $q=0.6$, 4) joint optimal tuning $q$ and $\xi$. $\alpha=3$, $\theta=0.8$, $R=3$.}
		\label{fig:optimal_q_alpha}
	\end{figure}
	\section{Conclusion}\label{conclusion}
	In this paper, we conducted an analytical study of optimizing AoI in a random access network by tuning system parameters such as the channel access probability and the packet arrival rate.
	Analytical expressions for the optimal peak AoI, as well as the corresponding system parameters, are derived for cases of seperate tuning and jointly tuning. In the seperate tuning case, when the node deployment density is small, information packets should be generated as frequently as possible, so as to achieve the optimal AoI performance. The same can apply to the optimal channel access rate, where transmitters should access the channel at each time slot. When the node deployment density becomes large, the optimal packet arrival rate and the optimal channel access probability should decrease as the node deployment density increases. In the jointly tuning case, in contrast, the optimal channel access probability is always set to be one and the optimal packet arrival rate shall decrease as the node deployment density increases. For all the cases of separately or jointly tuning of the channel access probability and the packet arrival rate, the optimal peak AoI linearly grows with the node deployment density as opposed to an exponential growth with the fixed channel access probability and the packet arrival rate. It is therefore of crucial importance to properly tune these parameters toward a satisfactory AoI performance especially in dense networks.
	\appendices
	\section{Proof of Lemma \ref{lemma:p}} \label{prooflemma1}
	Note that according to (9) and (10) in \cite{5226957}, the probability of successful transmission $p$ is determined by the following equation
	\begin{equation}\label{eq:Fixed-Eq1}
	p=\exp{\{-\lambda cR^2\rho q-\theta R^\alpha\gamma^{-1}\}},
	\end{equation}
	where $c=\frac{\pi\theta^{\frac{2}{\alpha}}}{\text{sinc}(\frac{2}{\alpha})}$
	and $\rho$ denotes the offered load of each transmitter. To derive the offered load $\rho$, let us define the state of each transmitter at time $t$ as the number of packets in the buffer at the begining of the time slot. As the buffer size of each transmitter is one, the state transition process of each transmitter can be model as a Markov chain with the state space $\textbf{X}=\{0,1\}$, where the transition matrix is given by
	\begin{equation}\label{eq:StateTrans}
	\textbf{P} =\left[\begin{array}{cc} p_{0,0} & p_{0,1} \\  p_{1,0} & p_{1,1} \\ \end{array}\right]=
	\left[\begin{array}{cc} 1-\xi+\xi qp & \xi(1- qp) \\  qp & 1- qp \\ \end{array}\right].
	\end{equation}
	where $ p_{i,j}$ is the probability of transiting $i\in \textbf{X}$ to state $j\in \textbf{X}$. According to \eqref{eq:StateTrans}, the steady-state distribution can be derived as
	\begin{equation}\label{eq:steady_pro}
	\left\{
	\begin{array}{lr}
	\pi_{0}=\frac{ qp}{\xi+ qp-\xi qp},   \\
	\pi_{1}=\frac{\xi-\xi qp}{\xi+ qp-\xi qp}.
	\end{array}
	\right.\end{equation}
	The offered load $\rho$ can then be written as
	\begin{equation}\label{eq:rho1}
	\rho=\frac{r}{qp},
	\end{equation}
	where $r$ is the effective packet arrival rate. As one incoming packet would be dropped if it sees a full buffer, the effective packet arrival rate $r$ is given by
	\begin{equation}\label{eq:effective_arrival}
	r=\xi\pi_0=\frac{\xi qp}{\xi+qp-\xi qp}.
	\end{equation}
	By combining \eqref{eq:rho1} and \eqref{eq:effective_arrival}, the offered load of each transmitter $\rho$ can be obtained as
	\begin{equation}\label{eq:rho}
	\rho=\frac{\xi}{\xi+qp-\xi qp}.
	\end{equation}
	Finally \eqref{eq:p} can be obtained by substituting \eqref{eq:rho} into \eqref{eq:Fixed-Eq1}.
	
	\section{Proof of Theorem 1}\label{studyProot}
	Let
	\begin{equation}
	f(p)= -\ln p-\frac{M}{N+p}-K,
	\label{eq:lnp}
	\end{equation}
	where
	\begin{equation}
	M=\lambda cR^2\frac{\xi}{1-\xi},~~
	N=\frac{\xi}{q(1-\xi)},~~
	K=\theta R^\alpha\gamma
	^{-1}.
	\label{MNK}
	\end{equation}
	It can be seen that $f(p)=0$ has the same non-zero roots as the fixed-point equation. The first derivative of $f(p)$ can be written as $f^{'}(p)=\frac{g(p)}{p(N+p)^2}$, where
	\begin{equation}
	g(p)=-\left(p+N-\frac{M}{2}\right)^2+\frac{{M}^2}{4}-MN.
	\label{eq:gp}
	\end{equation}
	Lemma \ref{lemma:gptfp} shows that the number of non-zero roots of $f(p)=0$ for $p\in(0,1]$ is crucially related with the number of non-zero roots of $g(p)=0$ for
	$p\in(0,1]$.
	
	\begin{lemma}
		$f(p)=0$ has three non-zero roots of $0<p_A<p_S<p_L<1$, if and only if g(p)=0 has two non-zero roots $0<p^{'}_1<p^{'}_2<1$ with $f(p_1^{'})<0$ and $f(p_2^{'})>0$; Otherwise, $f(p)=0$ has only one non-zero root $0<p_L<1$.
		\label{lemma:gptfp}
	\end{lemma}
	\begin{proof}
		Since $\lim_{p\to 0} f(p)>0$, $f(1)=-\frac{M}{N+1}-K<0$ and $f(p)$ is continuous function. According to the zero-point theorem, $f(p)$ has non-zero root. As $p(N+p)^2>0$ when $ 0<p\leq 1$, the number of non-zero roots of $f(p)=0$ for $p\in(0,1]$ is crucially related with the number of non-zero roots of $g(p)=0$ for $p\in(0,1]$. So, we consider the root of $g(p)$ in following scenarios:
		
		1) If $g(p)=0$ has no non-zero root for $p\in(0,1]$, then in this case $g(p)<0$ for $p\in (0,1]$. Thus, $f^{'}(p)<0$ and $f(p)$ decreases monotonically when $p\in(0,1]$, so $f(p)$ only has one non-zero root.
		
		2) If $g(p)=0$ has one non-zero root $0<p^{'}\leq1$, then, we have:
		
		~~a) If $g(p)=0$ has two roots when $p\in (-\infty,\infty)$, and one of them in range $p\in(0,1)$, we mark the root in (0,1) as $p^{'}$. So when $p\in (0,p^{'})$, $f(p)$ monotonically decrease and $p\in (p^{'},1)$, $f(p)$ monotonically increase, and we have $f(1)<0$, thus $f(p)$ only has one non-zero root.
		
		~~b) If $g(p)=0$ has one root when $p\in (-\infty,\infty)$, and this root in range $p\in(0,1)$, then $f(p)$ monotonically decrease when $p\in(0,1)$, so $f(p)$ only has one non-zero root. Combinging two cases, when $g(p)=0$ has one non-zero root, $f(p)=0$ only has one non-zero root.
		
		3) If $g(p)=0$ has two non-zero roots $0<p^{'}_1<p^{'}_2<1$, then we have:
		
		~~a) If $p_2^{'}=1$, then $g(p)<0$ for $p\in(0,p_1^{'})$ and $g(p)>0$ for $p\in(p_1^{'},1)$. As a result, $f^{'}(p)<0$ for $p\in(0,p_{1}^{'})$ and $f^{'}(p)>0$ for $p\in(p_{1}^{'},1)$, indicating that $f(p)$ monotonically decreases for $p\in(0,p_{1}^{'})$, and increases for $p\in(p_{1}^{'},1]$. Since $f(1)<0$, we can conclude that in this case, $f(p)=0$ only one non-zero root $0<p_L<1$.
		
		~~b) If $p_2^{'}<1$, then $g(p)<0$ for $p\in(0, p^{'}_1)\cup(p^{'}_2, 1)$, and $g(p)>0$ for $p\in (p^{'}_1, p^{'}_2)$. As a result, $f^{'}(p)<0$ for $p\in(0, p^{'}_1)\cup(p^{'}_2, 1)$, and $f^{'}(p)>0$ for $p\in (p^{'}_1, p^{'}_2)$, indicating that $f(p)$ monotonically decreases for $p\in(0, p^{'}_1)\cup(p^{'}_2, 1)$, and increases for $p\in (p^{'}_1, p^{'}_2)$.Then, we have if $f(p_{1}^{'})>0$ or $f(p_{2}^{'})<0$, $f(p)=0$ has one zero root
		$0<p_L\leq1$; otherwise, $f(p)=0$ has three non-zero roots $0<p_A<p_S<p_L\leq1$ in which $f(p_{1}^{'})<0$ or $f(p_{2}^{'})>0$.
	\end{proof}
	lemma \ref{lemma:3roots} further presents the necessary and sufficient condition for $g(p)=0$ has two non-zero roots $0<p^{'}_1<p^{'}_2<1$ with $f(p^{'}_{1})<0$ and $f(p^{'}_{2})>0$.
	\begin{lemma}
		$g(p)=0$ has two non-zero roots $0<p^{'}_1<p^{'}_2<1$ with $f(p_{1}^{'})<0$ and $f(p_{2}^{'})>0$ if and only if $\frac{4}{q}<\lambda cR^2<\frac{((1-\xi)q+\xi)^2}{q^2\xi(1-\xi)}$ and $\xi_l<\xi<\xi_h$, where $\xi_l$ and $\xi_h$ are given in \eqref{arrival_low}
		and \eqref{arrival_high},
		respectively.
		\label{lemma:3roots}
	\end{lemma}
	\begin{proof}
		$g(p)=0$ has two non-zero roots, when $\lim_{p\to 0} g(p)<0$, $g(1)< 0$ and peak value of $g(p)$ is larger than zero, and find that $g(p)=0$ has two non-zero roots $p^{'}_1$ and $p^{'}_2$, if
		\begin{equation}\frac{4}{q}<\lambda cR^2<\frac{((1-\xi)q+\xi)^2}{q^2\xi(1-\xi)},\text{and} \quad 0<q<1, 0<\xi<1.
		\label{eq:gpvalue}
		\end{equation}
		Two roots of $g(p)=0$ can be written as
		\begin{equation}p^{'}_1=\frac{M}{2}-N-\sqrt{\frac{{M}^2}{4}-MN},\end{equation}
		\begin{equation}
		p^{'}_2=\frac{M}{2}-N+\sqrt{\frac{{M}^2}{4}-MN},
		\end{equation}
		from $0<p^{'}_1<p^{'}_2<1$, we obtain
		\begin{align}
		\lambda cR^2<\frac{((1-\xi)q+\xi)^2}{q^2\xi(1-\xi)},
		\label{eq:gproots1}
		\\
		\frac{2}{q}<\lambda cR^2<\frac{2}{q}+\frac{2}{\xi}-2.
		\label{eq:gproots2}
		\end{align}
		So combining \eqref{eq:gpvalue}, \eqref{eq:gproots1}, \eqref{eq:gproots2}, the fixed-point equation has three non-zero roots $0<p_{A}<p_{S}<p_{L}<1$, if
		\begin{equation}\frac{4}{q}<\lambda cR^2<\min\left\{\frac{((1-\xi)q+\xi)^2}{q^2\xi(1-\xi)},\frac{2}{q}+\frac{2}{a}-2\right\}
		\label{eq:lambdacRcondition}\quad\text{and} \quad 0<q<1,0<\xi<1,\end{equation}
		\begin{align}
		f(p^{'}_1)&=-\frac{M}{\frac{{M}}{2}-\sqrt{\frac{{M}^2}{4}-MN}}-\ln\left(\frac{M}{2}-N-\sqrt{\frac{{M}^2}{4}-MN}\right)-\theta R^\alpha\gamma^{-1} < 0,
		\label{eq:fp1'MN}
		\\
		f(p^{'}_2)&=-\frac{M}{\frac{M}{2}+\sqrt{\frac{{M}^2}{4}-MN}}-\ln\left(\frac{M}{2}-N+\sqrt{\frac{{M}^2}{4}-MN}\right)-\theta R^\alpha\gamma^{-1} > 0,
		\label{eq:fp2'MN}
		\end{align}
		where $M=\lambda cR^2\frac{\xi}{1-\xi}$,
		$N=\frac{\xi}{q(1-\xi)}.$
		
		Firstly, we simplify \eqref{eq:lambdacRcondition}, which means
		$
		\frac{4}{q}<\frac{2}{q}+\frac{2}{\xi}-2$ and $
		\frac{4}{q}<\frac{((1-\xi)q+\xi)^2}{q^2\xi(1-\xi)}$.
		From that, we can get $q>\frac{\xi}{1-\xi}$. Then, we prove that $\frac{2}{q}+\frac{2}{\xi}-2>\frac{((1-\xi)q+\xi)^2}{q^2\xi(1-\xi)}$ in this case. We have
		\begin{equation}
		\frac{2}{q}+\frac{2}{\xi}-2-\frac{((1-\xi)q+\xi)^2}{q^2\xi(1-\xi)}=\frac{q^2(1-\xi)^2-\xi^2}{q^2\xi(1-\xi)},
		\end{equation}
		As $q>\frac{\xi}{1-\xi}$, we get $\xi<1-\frac{1}{1+q}$. And $q\in(0,1]$, thus, $\xi \in(0,0.5]$. Then, we get $q^2(1-\xi)^2>\xi^2$, and we have $\frac{2}{q}+\frac{2}{\xi}-2-\frac{((1-\xi)q+\xi)^2}{q^2\xi(1-\xi)}>0$, which means \eqref{eq:lambdacRcondition} can be written as
		\begin{equation}
		\frac{4}{q}<\lambda cR^2<\frac{((1-\xi)q+\xi)^2}{q^2\xi(1-\xi)}.
		\label{eq:lcr2final}
		\end{equation}
		
		Moreover, by combining \eqref{MNK}, \eqref{eq:fp1'MN}, \eqref{eq:fp2'MN}, we can obtain \eqref{arrival_low}, \eqref{arrival_high}.
	\end{proof}
	Finally, Theorem \ref{Theorem_p_root} can be obtained by combining Lemma \ref{lemma:gptfp} and Lemma \ref{lemma:3roots}.
	
	\section{Proof of monotonicity of $p_A$, $p_L$ }\label{studyPLtrend}
	According to \eqref{eq:p}, we have 	
	\begin{align}
	\frac{\partial{p}}{\partial{\xi}}&=\frac{\lambda cR^2p^2}{\lambda cR^2 p \xi(1-\xi)-(\frac{\xi}{q}+p(1-\xi))^2}=\frac{\lambda cR^2p^2}{g(p)}\frac{1}{(1-\xi)^2},
	\label{pvsalpha}
	\\
	\frac{\partial{p}}{\partial{q}}&=\frac{\lambda cR^2(\frac{\xi}{q})^2p}{\lambda cR^2 p \xi(1-\xi)-(\frac{\xi}{q}+p(1-\xi))^2}=\frac{\lambda cR^2(\frac{\xi}{q})^2p}{g(p)}\frac{1}{(1-\xi)^2},
	\label{pvsq}
	\\
	\frac{\partial{p}}{\partial{\lambda}}&=\frac{cR^2\frac{\xi}{q}(\xi+p q(1-\xi))}{\lambda cR^2 p \xi(1-\xi)-(\frac{\xi}{q}+p(1-\xi))^2}=\frac{cR^2\frac{\xi}{q}(\xi+p q(1-\xi))}{g(p)}\frac{1}{(1-\xi)^2}.
	\label{pvslambda}
	\end{align}
	where $g(p)$ is given in \eqref{eq:gp}. Let us consider the following scenarios:
	
	1) If $g(p)=0$ has no non-zero root for $p\in(0,1]$, then $g(p)<0$ for $p\in (0,1]$. In this case, \eqref{eq:p} has one-zero root $p_L$, which is a steady-state point according to the approximate trajectory analysis in \cite{6205590}. We then have have $g(p_L)<0$.
	
	2) If $g(p)=0$ has one non-zero root for $0<p^{'}_1\leq1$, then we have:
	
	~~a) If $g(p)=0$ has two roots when $p\in(-\infty,\infty]$, and one of them in range $p\in(0,1]$, then $g(p)<0$ for $p\in(0,p^{'}_1)$ and $g(p)>0$ for $p\in(p^{'}_1,1)$. \eqref{eq:p} has one-zero root $p_L$, and $p_L< p^{'}_1$, which is a steady-state point according to the approximate trajectory analysis in \cite{6205590}. We then have $g(p_L)<0$.
	
	~~b) If $g(p)=0$ has one root when $p\in (-\infty,\infty)$, and this root in range $p\in(0,1)$, then $g(p)<0$ for $p\in (0,1]$. \eqref{eq:p} has one-zero root $p_L$, which is a steady-state point according to the approximate trajectory analysis in \cite{6205590}. We then have $g(p_L)<0$.
	
	3) If $g(p)=0$ has two non-zero roots for $0<p^{'}_1<p^{'}_2<1$, then we have:
	
	~~a) If $p^{'}_2=1$, then $g(p)<0$ for $p\in(0,p^{'}_1)$ and $g(p)>0$ for $p\in(p^{'}_1,1)$, \eqref{eq:p} has one non-zero root $p_L<p^{'}_1$, which is a steady-state point according to the approximate trajectory analysis in \cite{6205590}. We then have $g(p_L)<0$.
	
	~~b) If $p^{'}_2<1$, then $g(p)<0$ for $p\in(0, p^{'}_1)\cup(p^{'}_2, 1)$, and $g(p)>0$ for $p\in (p^{'}_1, p^{'}_2)$. In this case, \eqref{eq:p} may has one non-zero root $p_L\in(0, p^{'}_1)\cup(p^{'}_2, 1)$, or three non-zero roots $p_A<p^{'}_1<p_S<p^{'}_2<p_L$, among which $p_A$ and $p_L$ are the steady-state point according to the approximate trajectory analysis in \cite{6205590}. We then have $g(p_A)<0$ and $g(p_L)<0$.
	
	Finally, we can conclude that $g(p_A)<0$, $g(p_L)<0$. Thus, it can be obtained from \eqref{pvsalpha}, \eqref{pvsq} and \eqref{pvslambda} that $\frac{\partial{p}}{\partial{\xi}}<0$, $\frac{\partial{p}}{\partial{q}}<0$, $\frac{\partial{p}}{\partial{\lambda}}<0$ at both the steady states $p_A$, $p_L$.
	
	\section{Proof of Theorem \ref{PAoIdef}}\label{proofPAoI}
	The peak AoI can be written as
	\begin{equation}
	A_p=E[T_{k-1}]+E[Y_k].
	\end{equation}
	where $T_{k-1}$ is service time of $(k-1)^{th}$ packet, $Y_k$ is the inter-departure time, between the $(k-1)^{th}$ packet served and the $k^{th}$ packet served.
	Thus, $E[T_{k-1}]$  can be obtained as
	\begin{equation}
	E[T_{k-1}]=\frac{1}{ qp}.
	\label{eq:Tk-1}
	\end{equation}
	For $Y_k$, it is a summation of the time interval $Y_k^a$  between $(k-1)^{th}$ packet departure and $k^{th}$ packet departure, and the time interval $Y_k^s$ between $k^{th}$ packet arrival and departure, i.e., $Y_k=Y_k^a+Y_k^s$.  As Fig \ref{AoIcurve} illustrates, the departure of ${(k-1)}^{th}$ packet and the arrival of $k^{th}$ packet can occur almost at the same time, but the time interval between the $k^{th}$ packet's arrival and departure costs at least one time slot.
	Thus, we have, $E[Y_k^a]=\frac{1}{\xi}-1$ and $E[Y_k^s]=\frac{1}{ qp}$.
	which leads to
	\begin{equation}
	E[Y_k]=E[t_a]+E[t_s]=\frac{1}{\xi}-1+\frac{1}{ qp}.
	\end{equation}
	
	\section{Proof of Theorem \ref{Theorem_OPtimalQ}}\label{prooftheorem2}
	According to \eqref{eq:p} and \eqref{eq:PeakAge}, we have
	\begin{equation}
	\frac{\partial{A_p}}{\partial{q}}=\frac{-2(p+q\frac{\partial{p}}{\partial{q}})}{q^2 p^2 }=\frac{2\xi^2}{q^2p}\left(-\frac{1}{\xi^2}-\frac{\lambda cR^2\frac{1}{q}}{\lambda cR^2 p \xi(1-\xi)-(\frac{\xi}{q}+p(1-\xi))^2}\right).\label{eq:Apvsq}
	\end{equation}
	We then have
	\begin{equation}
	\lim_{q \to 0}\frac{\partial{A_p}}{\partial{q}}<0
	\end{equation}
	and
	\begin{equation}
	\lim_{q\to 1}\frac{\partial{A_p}}{\partial{q}}=-\frac{2}{p_{*}}-\frac{2\lambda cR^2\xi^2\frac{1}{p_{*}}}{\lambda cR^2p_{*}\xi(1-\xi)-(\xi+p_{*}(1-\xi))^2},
	\end{equation}
	where $p_{*}$ is the non-zero root of \eqref{eq:q1valuep}. As $\frac{\partial{p}}{\partial{q}}<0$, we have
	\begin{equation}
	\lambda cR^2p_{*}\xi(1-\xi)-(\xi+p_{*}(1-\xi))^2<0.
	\end{equation}
	Then, when
	\begin{equation}
	\lambda c R^2>\frac{(\xi+p_{*}(1-\xi))^2}{\xi^2+p_{*}\xi(1-\xi)}=1+\frac{p_{*}(1-\xi)}{\xi},
	\end{equation}
	we have $\lim_{q\to 1}\frac{\partial{A_p}}{\partial{q}}>0$.
	The peak AoI $A_p$ can then be optimized when $q \in (0,1)$. By combining $\frac{\partial{A_p}}{\partial{q}}=0$ and \eqref{eq:p}, the optimal channel access probability $q$ can be obtained as
	\begin{equation}
	q=\frac{1}{\lambda cR^2-\exp{\left\{-\theta R^\alpha\gamma^{-1}-1\right\}}\frac{1-\xi}{\xi}}.
	\label{eq:accessrate}
	\end{equation}
	The optimal peak AoI can be obtained by substituting \eqref{eq:accessrate} into \eqref{eq:PeakAge}.
	
	When $	\lambda c R^2\leq 1+\frac{p_{*}(1-\xi)}{\xi}$, on the other hand, the optimal channel access rate is given by $q=1$, and the corresponding optimal peak AoI can be obtained by combining $q=1$ and \eqref{eq:PeakAge}.
	
	\section{Proof of Theorem \ref{Theorem_OPtimalAlpha}}\label{prooftheorem3}
	According to \eqref{eq:p} and \eqref{eq:PeakAge}, we have
	\begin{equation}\label{eq:opalpha}
	\frac{\partial{A_p}}{\partial{\xi}}=-\left(\frac{1}{\xi^2}+\frac{2\frac{\partial{p}}{\partial{\xi}}}{q p^2}\right)=-\frac{1}{\xi^2}-\frac{2\lambda cR^2\frac{1}{q}}{\lambda cR^2 p \xi(1-\xi)-(\frac{\xi}{q}+p(1-\xi))^2}
	\end{equation}
	We then have
	\begin{equation}
	\lim_{\xi \to 0}\frac{\partial{A_p}}{\partial{\xi}}<0
	\end{equation}
	and
	\begin{equation}
	\lim_{\xi \to 1}\frac{\partial{A_p}}{\partial{\xi}}=2\lambda c R^2 q-1.
	\end{equation}
	When
	$\lambda c R^2>\frac{1}{2 q}$, we have $\lim_{\xi \to 1}\frac{\partial{A_p}}{\partial{\xi}}>0$
	the peak AoI $A_p$ can then be optimized when $\xi \in (0,1)$. By combining $\frac{\partial{A_p}}{\partial{\xi}}=0$ and \eqref{eq:p}, the optimal packet arrival rate $\xi$ can be obtained as
	\begin{equation}
	\label{eq:arrivalrate}
	\xi=\frac{2q \exp{\left\{-\frac{2}{\sqrt{1+\frac{4}{q\lambda cR^2}}+1}-\theta R^\alpha \gamma^{-1}\right\}}}{q\lambda cR^2\left(\sqrt{1+\frac{4}{q\lambda cR^2}}+1\right)+2q\exp{\left\{-\frac{2}{\sqrt{1+\frac{4}{q\lambda cR^2}}+1}-\theta R^\alpha \gamma^{-1}\right\}}-2}.
	\end{equation}
	The optimal peak AoI can be obtained by substituting \eqref{eq:arrivalrate} into \eqref{eq:PeakAge}.
	
	When $\lambda cR^2\leq\frac{1}{2q}$, on the other hand, the optimal packet arrival rate is given by $\xi=1$, and the corresponding optimal peak AoI can be obtained by combining $\xi=1$ and \eqref{eq:PeakAge}.
	
	\section{Proof of Theorem \ref{Theorem_OPtimalqAlphaboth}}\label{prooftheorem4}
	We denote that
	\begin{equation}
	A_p^{*}=\underset{\{q\}}{\min}~A_{p}^{\xi=\xi_q^{*}}.
	\end{equation}
	From Theorem \ref{Theorem_OPtimalAlpha}, for $q<\frac{1}{2\lambda c R^2}$, we have
	\begin{equation}
	\frac{\mathrm{d}{A_{p}^{\xi=\xi_q^{*}}}}{\mathrm{d}{q}}=\frac{2(\lambda cR^2-\frac{1}{q})}{q}\exp{\left\{\lambda cR^2q+\theta R^\alpha \gamma^{-1}\right\}}.
	\end{equation}
	As $q<\frac{1}{2\lambda c R^2}$, thus, $\frac{\mathrm{d}{A_{p}^{\xi=\xi_q^{*}}}}{\mathrm{d}{q}}<0$, which means that $A_{p}^{\xi=\xi^{*}_{q}}$ decrease monotonically when $q<\frac{1}{2\lambda c R^2}$.
	On the other hand, for $q>\frac{1}{2\lambda c R^2}$, we denote $k=-\frac{2}{\sqrt{1+\frac{4}{q\lambda cR^2}}+1}$, $k\in(-1,-\frac{1}{2})$, then
	\begin{equation}
	A_{p}^{\xi=\xi_q^{*}}=\frac{1}{qe^{k}e^{-\theta R^\alpha \gamma^{-1}}(k+1)}.
	\end{equation}
	We find that $k$ decreases with the increase of $q$, $\frac{1}{e^{k}(k+1)}$ decreases with the increase of $k$ and $\frac{1}{e^{k}(k+1)}>0$, so $A_{p}^{\xi=\xi_q^{*}}$ decrease monotonically with the increase of $q$ when $q>\frac{1}{2\lambda c R^2}$.
	
	Moreover, when $\lambda cR^2=\frac{1}{2q}$, $A_{p}^{\xi_q^{*}=1}=A_{p}^{\xi_q^{*}<1}$, which means $A_{p}^{\xi=\xi_q^{*}}$ is a coutinuous function for $q$, thus $A_{p}^{\xi=\xi_q^{*}}$ decrease monotonically for $q\in (0,1]$.
	
	Therefore, $A_p^{*}=A_{p}^{\xi=\xi_q^{*}}$ when $q=1$ and Theorem \ref{Theorem_OPtimalqAlphaboth} can be obtained by combining $q=1$ and Theorem \ref{Theorem_OPtimalAlpha}.

	\bibliography{reference}
\end{document}